\documentclass[11pt,english]{article}
\usepackage{amsthm}
\usepackage{amsmath}
\usepackage{amssymb}
\usepackage{hyperref}
\hypersetup{colorlinks=true,linkcolor=black,citecolor=blue}
\PassOptionsToPackage{ocgcolorlinks}{hyperref}
\usepackage{fullpage}
\usepackage{nicefrac}
\usepackage{enumerate}
\usepackage{combelow}
\usepackage{xspace}
\usepackage[nottoc]{tocbibind}
\usepackage{tikz}
\usetikzlibrary{arrows.meta}
\usepackage{mathpazo}

\usepackage{multirow}

\makeatletter

\theoremstyle{plain}
\newtheorem{thm}{\protect\theoremname}[section]
\newtheorem{theorem}[thm]{Theorem}%[section]
\newtheorem{lem}[thm]{Lemma}%[section]
\newtheorem{definition}[thm]{Definition}%[section]
\newtheorem{corollary}[thm]{Corollary}%[section]
\newtheorem{fact}[thm]{Fact}
\newtheorem{proposition}[thm]{Proposition}
\newtheorem{question}[thm]{Open Question}

\makeatletter
\newcommand\slarge{\@setfontsize\slarge{14}{14}}
\makeatother

  \theoremstyle{definition}
  \newtheorem{defn}[thm]{\protect\definitionname}%[section]

\usepackage{complexity,color,bm}
\usepackage{colortbl}
\newcommand{\ignore}[1]{}%
\renewcommand{\SAT}{\textsf{CNF-SAT}\xspace}
\newcommand{\CNF}{\textsf{CNF}\xspace}

\newcommand{\eq}{{\sf Equality}\xspace}
\newcommand{\multeq}{{\sf Multi-Equality}\xspace}
\newcommand{\seleq}{{\sf Selective-Equality}\xspace}

\newcommand{\SZ}{\textsc{SumZero}}
\newcommand{\SEp}{\textsc{$\mathbb Z_p$-SumZero}}
\newcommand{\EQU}{\textsc{Eq}}
\newcommand{\EQUN}{\textsc{SelEq}}
\newcommand{\EQ}{\textsc{MultEq}}
\newcommand{\tEQ}{\text{\textsc{MultEq}}}
\newcommand{\SD}{\textsc{Disj}}
\newcommand{\cA}{\mathcal{A}}

\newcommand{\cF}{\mathcal{F}}
\newcommand{\cZ}{\mathcal{Z}}
\newcommand{\SMP}{\textsf{SMP}\xspace}
\renewcommand{\PCP}{\textsf{PCP}\xspace}

\newcommand{\cC}{\mathcal{C}}
\newcommand{\tS}{\widehat{S}}

\newcommand{\N}{\mathbb{N}}
\newcommand{\Z}{\mathbb{Z}}

\theoremstyle{plain}
\newtheorem{hypothesis}{Hypothesis}

\DeclareMathOperator*{\var}{var}

\DeclareMathOperator*{\ext}{ext}

\makeatother

\providecommand{\definitionname}{Definition}
\providecommand{\theoremname}{Theorem}

\renewcommand{\W}{\mbox{\sf W}\xspace}
\renewcommand{\FPT}{\mbox{\sf FPT}\xspace}
\newcommand{\SETH}{\mbox{\sf SETH}\xspace}
\newcommand{\ETH}{\mbox{\sf ETH}\xspace}
\newcommand{\gapETH}{\mbox{\sf Gap-ETH}\xspace}

\newcommand{\threesum}{\mbox{\sf 3-SUM}\xspace}
\newcommand{\ksum}{\mbox{\sf $k$-SUM}\xspace}
\newcommand{\subsetsum}{\mbox{\sf SUBSET-SUM}\xspace}

\newcommand{\mip}{\mbox{\sf MIP}\xspace}
\newcommand{\sumz}{\mbox{\sf Sum-Zero}\xspace}
\newcommand{\sumzp}{\mbox{\sf Sum-Zero($\mathbb{Z}_p$)}\xspace}
\newcommand{\Wone}{\mbox{\sf $\W[1]\neq \FPT$}\xspace}

\newcommand{\domset}{\mbox{\sf DomSet}\xspace}

\newcommand{\maxcover}{\mbox{\sf MaxCover}\xspace}

\newcommand{\minlabel}{\mbox{\sf MinLabel}\xspace}
\newcommand{\PSP}{\mbox{\sf PSP}\xspace}

\newcommand{\clique}{\mbox{\sf Clique}\xspace}

\title{\textbf{On the Parameterized Complexity of  \\
Approximating Dominating Set}}

\date{}

\author{
  Karthik C.\ S.\thanks{
    Weizmann Institute of Science, Israel.
    Email: \texttt{karthik.srikanta@weizmann.ac.il}}
\and 
  Bundit Laekhanukit\thanks{
    Shanghai University of Finance and Economics, China \& 
    Simons Institute for the Theory of Computing, USA.
    Email: \texttt{blaekhan@mpi-inf.mpg.de}
  }
\and
  Pasin Manurangsi\thanks{
    University of California, Berkeley, USA.
    Email: \texttt{pasin@berkeley.edu}
  }
}

\begin{document}

\maketitle
\vspace{-25pt}

\begin{abstract}
We study the parameterized complexity of approximating the $k$-Dominating Set (\domset) problem where an integer $k$ and a graph $G$ on $n$ vertices are given as input, and the goal is to find a dominating set of size at most $F(k) \cdot k$ whenever the graph $G$ has a dominating set of size $k$. When such an algorithm runs in time $T(k) \cdot \poly(n)$ (i.e., FPT-time) for some computable function $T$, it is said to be an \emph{$F(k)$-FPT-approximation algorithm} for $k$-\domset. Whether such an algorithm exists is listed in the seminal book of Downey and Fellows (2013) as one of the ``most infamous'' open problems in Parameterized Complexity. This work gives an almost complete answer to this question by showing the non-existence of such an algorithm under \Wone and further providing tighter running time lower bounds under stronger hypotheses. Specifically, we prove the following for every computable functions $T, F$ and every constant $\varepsilon > 0$:
\begin{itemize}
\item Assuming \Wone, there is no \emph{$F(k)$-FPT-approximation algorithm} for $k$-\domset.
\item Assuming the Exponential Time Hypothesis (\ETH), there is no $F(k)$-approximation algorithm for $k$-\domset that runs in $T(k) \cdot n^{o(k)}$ time.
\item Assuming the Strong Exponential Time Hypothesis (\SETH), for every integer $k \geq 2$, there is no $F(k)$-approximation algorithm for $k$-\domset that runs in $T(k) \cdot n^{k - \varepsilon}$ time.
\item Assuming the \ksum Hypothesis, for every integer $k \geq 3$, there is no $F(k)$-approximation algorithm for $k$-\domset that runs in $T(k) \cdot n^{\lceil k/2 \rceil - \varepsilon}$ time.
\end{itemize}

%In fact, the inapproximability ratios we achieve are even stronger than stated. Namely, for $\W[1]$, \ETH and \ksum-hardness results, the ratios can be improved to $(\log n)^{1/\poly(k)}$ whereas, for \SETH, the ratio is $(\log n)^{1/h(k)}$ for some function $h$.

Previously, only constant ratio FPT-approximation algorithms were ruled out under \Wone and $(\log^{1/4 - \varepsilon} k)$-FPT-approximation algorithms were ruled out under \ETH [Chen~and~Lin, FOCS 2016]. Recently, the non-existence of an $F(k)$-FPT-approximation algorithm for any function $F$ was shown under \gapETH [Chalermsook~et~al., FOCS 2017]. Note that, to the best of our knowledge, no running time lower bound of the form $n^{\delta k}$ for any absolute constant $\delta > 0$ was known before even for any constant factor inapproximation ratio.

Our results are obtained by establishing a connection between 
communication complexity and hardness of approximation, generalizing the ideas from a recent breakthrough work of Abboud~et~al. [FOCS 2017]. Specifically, we show that to prove hardness of approximation of a certain parameterized variant of the label cover problem, it suffices to devise a specific protocol for a communication problem that depends on which hypothesis we rely on. Each of these communication problems turns out to be either a well studied problem or a variant of one; this allows us to easily apply known techniques to solve them.

\end{abstract}

\clearpage
%\tableofcontents
%\renewcommand{\thefootnote}{\fnsymbol{footnote}}

\section{Introduction}
\label{sec:intro}

%The {\em dominating set} problem (\domset) is a classical problem in the area of combinatorial optimization and theoretical computer science.
In the {\em dominating set} problem (\domset), we are given an undirected graph $G$ on $n$ vertices and an integer $k$, and the goal is to decide whether there is a subset of vertices $S\subseteq V(G)$ of size $k$ such that every vertex outside $S$ has a neighbor in $S$ (i.e., $S$ {\em dominates} every vertex in $G$ and is thus called a {\em dominating set}). 
Often regarded as one of the classical problems in computational complexity, \domset was first shown to be NP-complete in the seminal work of Karp~\cite{Karp72}\footnote{To be precise, Karp showed NP-completeness of Set Cover, which is well-known to be equivalent to \domset.}. Thus, its optimization variant, namely the {\em minimum dominating set}, where the goal is to find a dominating set of smallest possible size, is also  NP-hard. To circumvent this apparent intractability of the problem, the study of an approximate version was initiated. The quality of an approximation algorithm is measured by the \emph{approximation ratio}, which is the ratio between the size of the solution output by an algorithm and the size of the minimum dominating set. A simple greedy heuristic for the problem, which has by now become one of the first approximation algorithms taught in undergraduate and graduate algorithm courses, was intensively studied and was shown to yield a $(\ln n - \ln \ln n + \Theta(1))$-approximation for the problem~\cite{Johnson74a,Chvatal79,Lovasz75,Srinivasan95,Slavik96}. On the opposite side, a long line of works in hardness of approximation~\cite{LundY94,RazS97,Feige98,AlonMS06,Moshkovitz15} culminated in the work of Dinur and Steurer~\cite{DinurS14}, who showed that obtaining a $(1 - \varepsilon)\ln n$-approximation for the problem is NP-hard for every $\varepsilon > 0$. This essentially settles the approximability of the problem.

Besides approximation, another widely-used technique to cope with NP-hardness is \emph{parameterization}. The parameterized version of \domset, which we will refer to simply as $k$-\domset, is exactly the same as the original decision version of the problem except that now we are not looking for a polynomial time algorithm but rather a {\em fixed parameter tractable} (FPT) algorithm -- one that runs in time $T(k)\cdot \poly(n)$ for some computable function $T$ (e.g., $T(k) = 2^k$ or $2^{2^k}$). Such running time will henceforth be referred to as \emph{FPT time}. Alas, even with this relaxed requirement, $k$-\domset still remains intractable: in the same work that introduced the $\W$-hierarchy, Downey and Fellows~\cite{DowneyF95} showed that $k$-\domset is complete for the class $\W[2]$, which is generally believed to not be contained in \FPT, the class of fixed parameter tractable problems. In the ensuing years, stronger running time lower bounds have been shown for $k$-\domset under strengthened assumptions. Specifically, Chen et al.~\cite{ChenHKX06} ruled out $T(k) \cdot n^{o(k)}$-time algorithm for $k$-\domset assuming the {\em Exponential Time Hypothesis} (\ETH)\footnote{\ETH~\cite{IP01} states that no subexponential time algorithm can solve 3-\SAT; see Hypothesis~\ref{hyp:eth}.}. Furthermore, P\u{a}tra\cb{s}cu and Williams \cite{PatrascuW10} proved, for every $k \geq 2$, that, under the {\em Strong Exponential Time Hypothesis} (\SETH)\footnote{\SETH~\cite{IP01}, a strengthening of \ETH, states that, for every $\varepsilon > 0$, there exists $k = k(\varepsilon) \in \mathbb{N}$ such that no $O(2^{(1 - \varepsilon) n})$-time algorithm can solve $k$-\SAT; see Hypothesis~\ref{hyp:seth}.}, not even $O(n^{k - \varepsilon})$ algorithm exists for $k$-\domset for any $\varepsilon > 0$. Note that the trivial algorithm that enumerates through every $k$-size subset and checks whether it forms a dominating set runs in $O(n^{k + 1})$ time. It is possible to speed up this running time using fast matrix multiplication~\cite{EisenbrandG04,PatrascuW10}. In particular, P\u{a}tra\cb{s}cu and Williams~\cite{PatrascuW10} themselves also gave an $n^{k + o(1)}$-time algorithm for every $k \geq 7$, painting an almost complete picture of the complexity of the problem.

Given the strong negative results for $k$-\domset discussed in the previous paragraph, it is natural to ask whether we can somehow incorporate the ideas from the area of approximation algorithms to come up with a \emph{fix parameter approximation (FPT-approximation) algorithm} for $k$-\domset. To motivate the notion of FPT-approximation algorithms, first notice that the seemingly reasonable $O(\log n)$-approximation given by the greedy algorithm can become unjustifiable when the optimum $k$ is small since it is even possible that the overhead paid is unbounded in terms of $k$. As a result, FPT-approximation algorithms require the approximation ratios to be bounded only in terms of $k$; specifically, for any computable function $F$, we say that an algorithm is an $F(k)$-FPT-approximation algorithm for $k$-\domset if it runs in FPT time and, on any input $(G, k)$ such that the minimum dominating set of $G$ is of size at most $k$, it outputs a dominating set of size at most $F(k) \cdot k$. 

This brings us to the main question addressed in our work:
%\begin{question}
\emph{Is there an $F(k)$-FPT-approximation algorithm for $k$-\domset for some computable function $F$?}
%\end{question}
This question, which dates back to late 1990s (see, e.g., \cite{DowneyFM06}), has attracted significant attention in literature~\cite{DowneyFM06,ChenGG06,DowneyFMR08,CaiH10,DowneyF13,HajiaghayiKK13,ChitnisHK13,BonnetE0P15,ChenL16,ChalermsookCKLMNT17}. In fact, it is even listed in the seminal textbook of Downey and Fellows~\cite{DowneyF13} as one of the six ``most infamous'' open questions\footnote{Since its publication, one of the questions, the parameterized complexity of biclique, has been resolved~\cite{Lin15}.} in the area of Parameterized Complexity. While earlier attempts fell short of ruling out either $F(k)$ that is super constant or all FPT algorithms (see Section~\ref{sec:intro:previous-works} for more details), the last couple of years have seen significant progresses on the problem. In a remarkable result of Chen and Lin~\cite{ChenL16}, it was shown that no FPT-approximation for $k$-\domset exists for any constant ratio unless $\W[1] = \FPT$. They also proved that, assuming \ETH, the inapproximability ratio can be improved to $\log^{1/4 - \varepsilon} k$ for any constant $\varepsilon > 0$. Very recently, Chalermsook~et~al.~\cite{ChalermsookCKLMNT17} proved, under the {\em Gap Exponential Time Hypothesis} (\gapETH)\footnote{\gapETH~\cite{D16,MR16}, another strengthening of \ETH, states that no subexponential time algorithm can distinguish satisfiable 3-\CNF formulae from ones that are not even $(1 - \varepsilon)$-satisfiable for some $\varepsilon > 0$.}, that no $F(k)$-approximation algorithm for $k$-\domset exists for any computable function $F$. %, thereby negatively answered the question. 
Such non-existence of FPT-approximation algorithms is referred to in literature as the \emph{total FPT-inapproximability} of $k$-\domset.

Although Chalermsook~et~al.'s result on the surface seems to settle the parameterized complexity of approximating dominating set, several aspects of the result are somewhat unsatisfactory. First, while \gapETH may be plausible, it is quite strong and, in a sense, does much of the work in the proof. Specifically, \gapETH itself already gives the gap in hardness of approximation; once there is such a gap, it is not hard\footnote{One issue grossed over in this discussion is that of \emph{gap amplification}. While \gapETH gives some constant gap, Chalermsook~et~al. still needed to amplify the gap to arrive at total FPT-inapproximability. Fortunately, unlike the NP-hardness regime that requires Raz's parallel repetition theorem~\cite{Raz98}, the gap amplification step in~\cite{ChalermsookCKLMNT17}, while non-trivial, only involved relatively simple combinatorial arguments. (See~\cite[Theorem 4.3]{ChalermsookCKLMNT17}.)} to build on it and prove other inapproximability results. As an example, in the analogous situation in NP-hardness of approximation, once one inapproximability result can be shown, others follow via relatively simple gap-preserving reductions (see, e.g., \cite{PY91}). On the other hand, creating a gap in the first place requires the PCP Theorem~\cite{AS98,ALMSS98}, which involves several new technical ideas such as local checkability of codes and proof composition\footnote{Even in the ``combinatorial proof'' of the PCP Theorem~\cite{Dinur07}, many of these tools still remain in use, specifically in the alphabet reduction step of the proof.}. Hence, it is desirable to bypass \gapETH and prove total FPT-inapproximability under assumptions that do not involve hardness of approximation in the first place. Drawing a parallel to the theory of NP-hardness of approximation once again, it is imaginable that a success in bypassing \gapETH may also reveal a ``PCP-like Theorem'' for parameterized complexity.

An additional reason one may wish to bypass \gapETH for the total FPT-inapproximability of $k$-\domset is that the latter is a statement purely about parameterized complexity, so one expects it to hold under a standard parameterized complexity assumption. Given that Chen and Lin~\cite{ChenL16} proved $\W[1]$-hardness of approximating $k$-\domset to within any constant factor, a concrete question here is whether we can show $\W[1]$-hardness of approximation for every function $F(k)$:

\begin{question}
Can we base the total FPT-inapproximability of $k$-\domset on \Wone?
\end{question}

Another issue not completely resolved by~\cite{ChalermsookCKLMNT17} is the running time lower bound. While the work gives a quite  strong running time lower bound that rules out any $T(k) \cdot n^{o(k)}$-time $F(k)$-approximation algorithm for any computable functions $T$ and $F$, it is still possible that, say, an $O(n^{0.5 k})$-time algorithm can provide a very good (even constant ratio) approximation for $k$-\domset.
Given the aforementioned $O(n^{k - \varepsilon})$ running time lower bound for exact algorithms of $k$-\domset by P\u{a}tra\cb{s}cu and Williams~\cite{PatrascuW10}, it seems reasonable to ask whether such a lower bound can also be established for approximation algorithms:

\begin{question}
Is it hard to approximate $k$-\domset in $O(n^{k - \varepsilon})$-time?
\end{question}

This question has perplexed researchers, as even with the running time of, say, $O(n^{k-0.1})$, no $F(k)$-approximation algorithm is known for $k$-\domset for any computable function $F$.

%%%%%%%%%%%%%%%%%%%%%%%%%%%%%%%%%%%%%%%%%%
%%%%%%%%%%%%%%%%%%%%%%%%%%%%%%%%%%%%%%%%%%
%%%%%%%%%%%%%%%%%%%%%%%%%%%%%%%%%%%%%%%%%%
 
\subsection{Our Contributions}
\label{sec:intro:our-results}

Our contributions are twofold.
Firstly, at a higher level, we prove parameterized inapproximabilty results for $k$-\domset, answering the two aforementioned open questions (and more).
Secondly, at a lower level, we demonstrate a connection between 
communication complexity and parameterized inapproximability,
allowing us to translate running time lower bounds for parameterized problems into 
parameterized hardness of approximation. This latter part of the contribution extends ideas from a recent breakthrough of Abboud et al.~\cite{ARW17}, who discovered similar connections and used them to establish inapproximability for problems in P. In this subsection, we only focus on the first part of our contributions. The second part will be discussed in detail in Section~\ref{sec:comm-connection}.
%The latter is the first result of its kind.
%To the best of our knowledge, there were no known NP-hard optimization problems whose 
%approximation hardness were derived from polynomial-time solvable problems.

\subsubsection*{Parameterized Inapproximability of Dominating Set}
Our first batch of results are the inapproximability results for $k$-\domset under various standard assumptions in parameterized complexity and fine-grained complexity: \Wone, \ETH, \SETH and the \ksum Hypothesis. First, we show total inapproximability of $k$-\domset under \Wone. In fact, we show an even stronger\footnote{Note that the factor of the form $(\log n)^{1/\poly(k)}$ is stronger than that of the form $F(k)$. To see this, assume that we have an $F(k)$-FPT-approximation algorithm for some computable function $F$. We can turn this into a $(\log n)^{1/\poly(k)}$-approximation algorithm by first checking which of the two ratios is smaller. If $F(k)$ is smaller, then just run the $F(k)$-FPT-approximation algorithm. Otherwise, use brute force search to solve the problem. Since the latter case can only occur when $n \leq \exp(F(k)^{\poly(k)})$, we have that the running time remains FPT.} inapproximation ratio of $(\log n)^{1/\poly(k)}$:

\begin{theorem}\label{thm:domW1}
Assuming \Wone, no FPT time algorithm can approximate $k$-\domset  to within a factor of $(\log n)^{1/\poly(k)}$.
\end{theorem}

Our result above improves upon the constant factor inapproximability result of Chen and Lin~\cite{ChenL16} and resolves the question of whether we can base total FPT inapproximability of $k$-\domset on a purely parameterized complexity assumption. Furthermore, if we are willing to assume the stronger \ETH, we can even rule out all $T(k) \cdot n^{o(k)}$-time algorithms:

\begin{theorem}\label{thm:domETH}
Assuming \ETH, no $T(k) \cdot n^{o(k)}$-time algorithm can approximate $k$-\domset to within a factor of $(\log n)^{1/\poly(k)}$.
\end{theorem}

Note that the running time lower bound and approximation ratio ruled out by the above theorem are exactly the same as those of Charlermsook et al.'s result based on \gapETH~\cite{ChalermsookCKLMNT17}. In other words, we successfully bypass \gapETH from their result completely. Prior to this, the best known \ETH-based inapproximability result for $k$-\domset due to Chen and Lin \cite{ChenL16} ruled out only $(\log^{1/4 + \varepsilon} k)$-approximation for $T(k) \cdot n^{o(\sqrt{k})}$-time algorithms.

Assuming the even stronger hypothesis, \SETH, we can rule out $O(n^{k - \varepsilon})$-time approximation algorithms for $k$-\domset, matching the running time lower bound from~\cite{PatrascuW10} while excluding not only exact but also approximation algorithms. We note, however, that the approximation ratio we get in this case is not $(\log n)^{1/\poly(k)}$ anymore, but rather $(\log n)^{1/\poly(k,e(\varepsilon))}$ for some function $e$, which arises from \SETH and the Sparsification Lemma~\cite{IPZ01}.

\begin{theorem}\label{thm:domSETH}
There is a function $e: \mathbb{R}^+ \to \N$ such that, assuming \SETH, for every integer $k \geq 2$ and for every $\varepsilon > 0$, no $O(n^{k - \varepsilon})$-time algorithm can approximate $k$-\domset to within a factor of $(\log n)^{1/\poly(k,e(\varepsilon))}$.
\end{theorem}

Finally, to demonstrate the flexibility of our proof techniques (which will be discussed at length in the next section), we apply the framework to the \ksum Hypothesis\footnote{The \ksum Hypothesis states that, for every $\varepsilon > 0, k \in \mathbb{N}$ such that $k \geq 3$, no $O(n^{\lceil k/2 \rceil - \varepsilon})$-time algorithm solves the \ksum problem; see Hypothesis~\ref{hyp:ksumh}.} which yields an $n^{\lceil k/2 \rceil - \varepsilon}$ running time lower bound for approximating $k$-\domset as stated below.

\begin{theorem}\label{thm:domksum}
Assuming the \ksum Hypothesis, for every integer $k \geq 3$ and for every $\varepsilon > 0$, no $O(n^{\lceil k/2 \rceil - \varepsilon})$-time algorithm can approximate $k$-\domset to within a factor of $(\log n)^{1/\poly(k)}$.
\end{theorem}

We remark here that the \ksum problem is known to be $\W[1]$-hard~\cite{DowneyF95,AbboudLW14} and our proof of Theorem~\ref{thm:domksum} indeed yields an alternative proof of $\W[1]$-hardness of approximating $k$-\domset (Theorem~\ref{thm:domW1}). Nevertheless, we provide a different self-contained $\W[1]$-hardness reduction directly from \clique since the ideas there are also useful for our ETH-hardness result (Theorem~\ref{thm:domETH}).

%In this paper, we show under $\mathrm{FPT}\neq W[1]$ that $k$-$\domset$ is {\em total FPT-inapproximable}. Furthermore, we show that, assuming ETH and SETH, the \domset admits no $(\log n)^{1/k}$-approximation algorithm that runs in time $n^{o(k)}$ and $n^{k-\epsilon}$, for any $\epsilon>1$, respectively.
%This improves upon the $\mathrm{W}[1]$-hardness result by \cite{ChenL16} and bypasses Gap-ETH used in the work of Chalermsook~et~al~\cite{ChalermsookCKLMNT17}.
%Our technique allows us to extend the hardness results to be based on other popular conjecture, namely the $k$-SUM Hypothesis, which asserts that $k$-SUM admits no $n^{\lfloor k/2 \rfloor - \epsilon}$-time algorithm for any $\epsilon>0$. Under the $k$-SUM Hypothesis, we show that $k$-\domset admits no $n^{\lfloor k/2 \rfloor - \epsilon}$-time $(\log n)^{1/k}$-approximation algorithm.
%Henceforth, we settle the parameterized complexity of approximating $k$-\domset under almost popular conjectures. 

%\pnote{May be add some implications here? Some $\W[2]$-complete problems that we can rule out their approx algos?}

The summary of our results and those from previous works are shown in Table~\ref{tab:summary}.

%\begin{table}[h]
%\begin{center}
%\begin{tabular}{c}
%{\bf Summary of Previous Works}\\
%\begin{tabular}{c | c | c | c}
%Complexity Assumption & Inapproximability Ratio & Running Time Lower Bound & References\\
%\hline
%%%%% W[1] %%%%
%  \Wone 
%  & Any constant 
%  & $T(k) \cdot \poly(n)$
%  & \cite{ChenL16}\\
%
%%%%% ETH %%%%
%  \ETH
%  & $(\log k)^{1/4+\epsilon}$ 
%  & $T(k) \cdot n^{o(\sqrt{k})}$
%  & \cite{ChenL16}\\
%%%%% Gap-ETH %%%%
%  \gapETH%~\cite{D16,MR16}
%  & $(\log n)^{1/\poly(k)}$
%  & $T(k) \cdot n^{o(k)}$
%  & \cite{ChalermsookCKLMNT17}%\\
%%%%% ETH %%%%
%%  SETH
%%  & N/A 
%%  & N/A\\
%%%%% k-SUM %%%%
%%  $k$-SUM Hypothesis
%%  & N/A 
%%  & N/A\\
%\end{tabular}
%\\
%\\
%{\bf Our Results}\\
%\begin{tabular}{c| c | c}
%Complexity Assumption & Inapproximability Ratio & Running Time Lower Bound\\
%\hline
%%%%% W[1] %%%%
%  \Wone 
%  & $(\log n)^{1/\poly(k)}$ 
%  & $T(k) \cdot \poly(n)$\\
%%%%% ETH %%%%
%  \ETH
%  & $(\log n)^{1/\poly(k)}$ 
%  & $T(k) \cdot n^{o(k)}$\\
%%%%% SETH %%%%
%  \SETH
%  & $(\log n)^{1/\poly(k,e(\varepsilon))}$ 
%  & $O(n^{k-\epsilon})$\\
%%%%% k-SUM %%%%
%  \ksum Hypothesis
%  & $(\log n)^{1/\poly(k)}$ 
%  & $O(n^{\lceil k/2 \rceil - \epsilon})$\\
%\end{tabular}
%\end{tabular}
%\end{center}
%\caption{Summary of our and previous results on $k$-\domset. We only show those whose inapproximability ratios are at least some constant greater than one (i.e., we exclude additive inapproximability results). Here $e:\mathbb{R}^+ \to \N$ is some   function, $T: \N \to \N$ can be any computable function and $\varepsilon$ can be any positive constant.}
%\label{tab:summary}
%\end{table}

\begin{table}[h]
\begin{center}
\resizebox{\textwidth}{!}{
\begin{tabular}{c}
{\bf Summary of Previous Works and The Results in This Paper}\\
\begin{tabular}{c| c | c | c}
Complexity Assumption & Inapproximability Ratio & Running Time Lower Bound&Reference\\
\hline
%%%% W[1] %%%%
 \multirow{ 2}{*}{\Wone} 
  & Any constant 
  & $T(k) \cdot \poly(n)$
  & \cite{ChenL16}\\
  & $(\log n)^{1/\poly(k)}$ 
  & $T(k) \cdot \poly(n)$
  & This paper\\ \hline
%%%% ETH %%%%
  \multirow{ 2}{*}{\ETH}
  & $(\log k)^{1/4+\epsilon}$ 
  & $T(k) \cdot n^{o(\sqrt{k})}$
  & \cite{ChenL16}\\
  & $(\log n)^{1/\poly(k)}$ 
  & $T(k) \cdot n^{o(k)}$
  & This paper\\\hline
%%%% Gap-ETH %%%%
  \gapETH%~\cite{D16,MR16}
  & $(\log n)^{1/\poly(k)}$
  & $T(k) \cdot n^{o(k)}$
  & \cite{ChalermsookCKLMNT17}\\\hline
%%%% SETH %%%%
 \multirow{ 2}{*}{\SETH}
  & Exact 
  & $O(n^{k-\epsilon})$
  & \cite{PatrascuW10}\\
  & $(\log n)^{1/\poly(k,e(\varepsilon))}$ 
  & $O(n^{k-\epsilon})$
  & This paper\\\hline
%%%% k-SUM %%%%
  \ksum Hypothesis
  & $(\log n)^{1/\poly(k)}$ 
  & $O(n^{\lceil k/2 \rceil - \epsilon})$
  & This paper\\
\end{tabular}
\end{tabular}}
\end{center}
\caption{Summary of our and previous results on $k$-\domset. We only show those whose inapproximability ratios are at least some constant greater than one (i.e., we exclude additive inapproximability results). Here $e:\mathbb{R}^+ \to \N$ is some   function, $T: \N \to \N$ can be any computable function and $\varepsilon$ can be any positive constant.
The \gapETH has been bypassed in this paper, and prior to this paper, the \ksum Hypothesis had never been used in proving inapproximability of $k$-\domset.}
\label{tab:summary}
\end{table}

\subsection{Comparison to Previous Works}
\label{sec:intro:previous-works}

%Besides being NP-complete and $\mathrm{W}[2]$-hard, $k$-\domset is also known to admit no $n^{o(k)}$-time algorithm under the ETH by the work of Chen~et~al.~\cite{ChenHKX06} and no $n^{k-\epsilon}$-time algorithm, for any constant $\epsilon> 0$, by the work of P\u{a}tra\cb{s}cu and Williams \cite{PatrascuW10} under the SETH.

In addition to the lower bounds previously mentioned, the parameterized inapproximability of $k$-\domset has also been investigated in several other works~\cite{DowneyFMR08,ChitnisHK13,HajiaghayiKK13,BonnetE0P15}. Specifically, Downey et al.~\cite{DowneyFMR08} showed that obtaining an \emph{additive} constant approximation for $k$-\domset is $\W[2]$-hard. On the other hand, in~\cite{HajiaghayiKK13,ChitnisHK13}, the authors ruled out $(\log k)^{1 + \varepsilon}$-approximation in time $\exp(\exp((\log k)^{1 + \varepsilon}))\cdot\poly(n)$ for some fixed constant $\varepsilon>0$ by assuming \ETH and the {\em projection game conjecture} proposed in \cite{Moshkovitz15}. Further, Bonnet~et~al.~\cite{BonnetE0P15} ruled out $(1+\varepsilon)$-FPT-approximation, for some fixed constant $\varepsilon>0$, assuming \gapETH\footnote{The authors assume the same statement as \gapETH (albeit, with imperfect completeness) but have an additional assertion that it is implied by \ETH (see Hypothesis~1 in \cite{BonnetE0P15}). It is not hard to see that their assumption can be replaced by \gapETH.}. We note that, with the exception of $\W[2]$-hardness results~\cite{DowneyF95,DowneyFMR08}, our results subsume all other aforementioned lower bounds regarding $k$-\domset, both for approximation~\cite{ChitnisHK13,HajiaghayiKK13,BonnetE0P15,ChenL16,ChalermsookCKLMNT17} and exact algorithms~\cite{ChenHKX06,PatrascuW10}. 

%The parameterized complexity in the regime of approximation algorithms have been investigated in \cite{ChitnisHK13,HajiaghayiKK13,BonnetE0P15,ChenL16,ChalermsookCKLMNT17}.
%
%Hajiaghayi~et~al.~\cite{HajiaghayiKK13} ruled out $(\log k)^{\varepsilon}$-approximation, for some fixed constant $\varepsilon>0$, in time $\exp(\exp((\log k)^{\varepsilon-1}))$ by assuming the ETH and the {\em projection game conjecture} proposed in \cite{Moshkovitz15}.
%
%Bonnet~et~al.~\cite{BonnetE0P15} ruled out $(1+\varepsilon)$-approximation, for some fixed constant $\varepsilon>0$, by assuming the ETH and the linear-size version of the projection game conjecture.
%
%Chen-Lin\cite{ChenL16} ruled out constant approximation in %FPT-time under \Wone,
%and in the very recent result,
%Chalermsook~et~al.~\cite{ChalermsookCKLMNT17} they completely ruled out 
%any $F(k)$-approximation in FPT-time for any computable function $F$
%by assuming the Gap-ETH.

While our techniques will be discussed at a much greater length in the next section (in particular we compare our technique with \cite{ARW17} in Section~\ref{subsec:framework}), we note that our general approach is to first show inapproximability of a parameterized variant of the {\em Label Cover} problem called \maxcover and then reduce \maxcover to $k$-\domset. The first step employs the connection between communication complexity and inapproximability of \maxcover, whereas the second step follows directly from the reduction in \cite{ChalermsookCKLMNT17}
(which is in turn based on \cite{Feige98}). While \maxcover was not explicitly defined until~\cite{ChalermsookCKLMNT17}, its connection to $k$-\domset had been implicitly used both in the work of P\u{a}tra\cb{s}cu and Williams~\cite{PatrascuW10} and that of Chen and Lin~\cite{ChenL16}.

From this perspective, the main difference between our work and~\cite{PatrascuW10,ChenL16,ChalermsookCKLMNT17} is the source of hardness for \maxcover. Recall that P\u{a}tra\cb{s}cu and Williams~\cite{PatrascuW10} ruled out only exact algorithms; in this case, a relatively simple reduction gave hardness for the exact version of \maxcover. On the other hand, both Chalermsook et al.~\cite{ChalermsookCKLMNT17} and Chen and Lin~\cite{ChenL16} ruled out approximation algorithms, meaning that they needed gaps in their hardness results for \maxcover. Chalermsook et al. obtained their initial gap from their assumption (\gapETH), after which they amplified it to arrive at an arbitrarily large gap for \maxcover. On the other hand,~\cite{ChenL16} derived their gap from the hardness of approximating Maximum $k$-Intersection shown in Lin's earlier breakthrough work~\cite{Lin15}. Lin's proof~\cite{Lin15} made use of certain combinatorial objects called \emph{threshold graphs} to prove inapproximability of Maximum $k$-Intersection. Unfortunately, this construction was not very flexible, in the sense that it produced \maxcover instances with parameters that were not sufficient for proving total-FPT-inapproximability for $k$-\domset. Moreover, his technique (i.e., threshold graphs) was limited to reductions from $k$-\clique and was unable to provide a tight running time lower bound under \ETH. By resorting to the connection between \maxcover and communication complexity, we can generate \maxcover instances with wider ranges of parameters from much more general assumptions, allowing us to overcome the aforementioned barriers.

\paragraph*{Organization.} 
In the next section, we give an overview of our lower level contributions; for readers interested in the general ideas without too much notational overhead, this section covers most of the main ideas from our paper through a proof sketch of our $\W[1]$-hardness of approximation result (Theorem~\ref{thm:domW1}). After that, in Section~\ref{sec:prelim}, we define additional notations and preliminaries needed to formalize our proofs. Section~\ref{sec:typeA} provides a definition for Product Space Problems (\PSP) and rewrites the hypotheses in these terms. Next, in Section~\ref{sec:typeA:cc2maxcover}, we establish a general theorem converting communication protocols to a reduction from \PSP to \maxcover. Sections~\ref{sec:SETH},~\ref{sec:multeq} and~\ref{sec:ksum} provide communication protocols for our problems of interest: Set Disjointness, \multeq and \sumz. Finally, in Section~\ref{sec:conclusion}, we conclude with a few open questions and research directions.

\section{Connecting Communication Complexity and Parameterized Inapproximability: An Overview}
\label{sec:comm-connection}

This section is devoted to presenting our connection between communication complexity and parameterized inapproximability (which is one of our main contributions as discussed in the introduction) 
and serves as an overview for all the proofs in this paper.
As mentioned previously, our discovery of this connection is inspired by the work of Abboud~et~al.~\cite{ARW17} who showed the connection between the communication protocols and  hardness of approximation for problems in \P. More specifically, they showed how a Merlin-Arthur protocol for {\em Set Disjointness} with certain parameters implies the \SETH-hardness of approximation for a problem called \emph{\PCP-Vectors} and used it as  the starting point to prove inapproximability of other problems in \P.
We extend this idea by identifying a communication problem associated with each of the complexity assumptions (\Wone, \ETH, \SETH and \ksum Hypothesis) and then prove a generic theorem that translates communication protocols for these problems to conditional hardness of approximation for a parameterized variant of the {\em Label Cover} problem called \maxcover~\cite{ChalermsookCKLMNT17}.
Since the hardness of \maxcover is known to imply the hardness of $k$-\domset~\cite{ChalermsookCKLMNT17} (see Section~\ref{sec:dom-set}), we have arrived at our inapproximability results for $k$-\domset.
As the latter part is not the contribution of this paper, we will focus on explaining the connection between communication complexity and the hardness of approximating \maxcover. We start by defining \maxcover:

\begin{definition} The input for \maxcover is a \emph{label cover instance} $\Gamma$, which consists of a bipartite graph $G=(U, W; E)$ such that $U$ is partitioned into $U_1\cup \cdots \cup U_q$ and $W$ is partitioned into $W_1\cup \cdots \cup W_h$. We sometimes refer to $U_{i}$'s and $W_j$'s as left and right {\em super-nodes} of $\Gamma$, respectively.

A solution to \maxcover is called a {\em labeling},
which is a subset of vertices $S \subseteq W$ formed by picking a vertex $w_{j}$ from each $W_j$ (i.e., $|S\cap W_j| = 1$ for all $j \in [h]$). We say that a labeling $S$ {\em covers} a left super-node $U_{i}$ if there exists a vertex $u_{i} \in U_{i}$ such that $u_i$ is a neighbor of every vertex in $S$.
The goal in \maxcover is to find a labeling that covers the maximum fraction of left super-nodes.
\end{definition}

For concreteness, we focus on the $\W[1]$-hardness proof (Theorem~\ref{thm:domW1}); at the end of this subsection, we will discuss how this fits into a larger framework that encapsulates other hypotheses as well.

For the purpose of our current discussion, it suffices to think of \maxcover as being parameterized by $h$, the number of right super-nodes; from this viewpoint, we would like to show that it is $\W[1]$-hard to approximate \maxcover to within $(\log n)^{1/\poly(h)}$ factor. For simplicity, we shall be somewhat imprecise in our overview below, all proofs will be formalized later in the paper.

We reduce from the $k$-\clique problem, which is well-known to be $\W[1]$-hard~\cite{DowneyF95}. The input to $k$-\clique is an integer $k$ and a graph which we will call $G' = (V', E')$ to avoid confusion with the label cover graph. The goal is to determine whether $G'$ contains a clique of size $k$. Recall that, to prove the desired $\W[1]$-hardness, it suffices to provide an \emph{FPT-reduction} from any $k$-\clique instance $(G',k)$ to approximate \maxcover instance $G = (U, W; E)$; this is an FPT-time reduction such that the new parameter $h$ is bounded by a function of the original parameter $k$. Furthermore, since we want a hardness of approximation result for the latter, we will also show that, when $(G',k)$ is a YES instance of $k$-\clique, there is a labeling of $G$ that covers all the left super-nodes. On the other hand, when $(G',k)$ is a NO instance of $k$-\clique, we wish to show that every labeling of $G$ will cover at most $1/(\log n)^{1/\poly(h)}$ fraction of the left super-nodes. If we had such a reduction, then we would have arrived at the total FPT-inapproximability of \maxcover under \Wone. But, how would we come up with such a reduction? We will do this by devising a specific kind of protocol for a communication problem!

\subsection{A Communication Problem for $k$-\clique}
\label{sec:protocol-for-clique}

 The communication problem related to $k$-\clique we consider is a multi-party problem where there are $h = \binom{k}{2}$ players, each associated with a two-element subset $\{i, j\}$ of $[k]$. The players cannot communicate with each other. Rather, there is a referee that they can send messages to. Each player $\{i, j\}$ is given two vertices $u^{\{i, j\}}_i$ and $u^{\{i, j\}}_j$ such that $\{u^{\{i, j\}}_i, u^{\{i, j\}}_j\}$ forms an edge in $G'$. The vertices $u^{\{i, j\}}_i$ and $u^{\{i, j\}}_j$ are allegedly the $i$-th and $j$-th vertices of a clique respectively. The goal is to determine whether there is indeed a $k$-clique in $G'$ such that, for every $\{i, j\} \subseteq [k]$, $u^{\{i, j\}}_i$ and $u^{\{i, j\}}_j$ are the $i$-th and $j$-th vertices of the clique.

The communication protocol that we are looking for is a one-round protocol with public randomness and by the end of which the referee is the one who outputs the answer. Specifically, the protocol proceeds as follows. First, the players and the referee together toss $r$ random coins. Then, each player sends an $\ell$-bit message to the referee. Finally, the referee decides, based on the messages received and the randomness, either to accept or reject. The protocol is said to have perfect completeness and soundness $s$ if (1) when there is a desired clique, the referee always accepts and (2) when there is no such clique, the referee accepts with probability at most $s$. The model described here is referred to in the literature  as the multi-party {\em Simultaneous Message Passing} (\SMP) model \cite{Yao79,BGKL03,FOZ16}. 
We refer to a protocol in the \SMP model as an \SMP protocol.

\paragraph{From Communication Protocol to \maxcover.} Before providing a protocol for the previously described communication problem, let us describe how to turn the protocol into a label cover instance $G = (U = U_1 \cup \cdots \cup U_q, W = W_1 \cup \cdots \cup W_h; E)$.
\begin{itemize}
\item Let $h = \binom{k}{2}$. Again, we associate elements in $[h]$ with two-element subsets of $[k]$. Each right super-node $W_{\{i, j\}}$ represents Player $\{i, j\}$. Each vertex in $W_{\{i, j\}}$ represents a possible input to the player, i.e., we have one vertex $a_{\{u, v\}}\in W_{\{i, j\}}$ for each edge $\{u, v\} \in E'$ in the graph $G'$. Assume w.l.o.g. that $i < j$ and $u < v$. This vertex $a_{\{u, v\}}$ represents player $\{i, j\}$ receiving $u$ and $v$ as the alleged $i$-th and $j$-th vertices of the clique respectively. %; one vertex $a_{(u, v)}$ representing player $\{i, j\}$ receiving $u, v$ as alleged $i$-th and $j$-th vertex of the clique respectively and the other vertex $a_{(v, u)}$ representing player $\{i, j\}$ receiving $v, u$ as alleged $i$-th and $j$-th vertex of the clique.
\item Let $q = 2^r$. We associate each element in $[q]$ with an $r$-bit string. For each $\gamma \in \{0, 1\}^r$, the left super-node $U_{\gamma}$ contains one node corresponding to each \emph{accepting configuration} on randomness $\gamma$; that is, for each $h$-tuple of $\ell$-bit strings $(m_{\{1, 2\}}, \dots, m_{\{k - 1, k\}}) \in (\{0, 1\}^\ell)^h$, there is a node $(m_{\{1, 2\}}, \dots, m_{\{k - 1, k\}})$ in $U_\gamma$ iff the referee on randomness $\gamma$ and  message $m_{\{1, 2\}}, \dots, m_{\{k - 1, k\}}$ from all the players accepts.
\item The edges in $E$ are defined as follows. Recall that each node $a$ in a right super-node $W_{\{i, j\}}$ corresponds to an input that each player receives in the protocol. For each $\gamma \in \{0, 1\}^r$, suppose that the message produced on this randomness by the $\{i, j\}$-th player on the input corresponding to $a$ is $m^{a, \gamma}$. We connect $a$ to every accepting configuration on randomness $\gamma$ that agrees with the message $m^{a, \gamma}$. More specifically, for every $\gamma \in \{0, 1\}^r$, $a$ is connected to every vertex $(m_{\{1, 2\}}, \dots, m_{\{k - 1, k\}}) \in U_\gamma$ iff $m_{\{i, j\}} = m^{a, \gamma}$.
\end{itemize}
Consider any labeling $S \subseteq W$. It is not hard to see that, if we run the protocol where the $\{i, j\}$-th player is given the edge corresponding to the unique element in $S \cap W_{\{i, j\}}$ as an input, then the referee accepts a random string $\gamma \in \{0, 1\}^r$ if and only if the left super-node $U_{\gamma}$ is covered by the labeling $S$. In other words, the fraction of the left super-nodes covered by $S$ is exactly equal to the acceptance probability of the protocol. This means that if $(G',k)$ is a YES-instance of $k$-\clique, then we can select $S$ corresponding to the edges of a $k$-clique and every left super-node will be covered. On the other hand, if $(G',k)$ is a NO-instance of $k$-\clique, there is no subset $S$ that corresponds to a valid $k$-clique, meaning that every labeling $S$ covers at most $s$ fraction of the edges. Hence, we have indeed arrived at hardness of approximation for \maxcover. Before we move on to describe the protocol, let us note that the running time of the reduction is $\poly(2^{r + \ell h}, |E'|)$, which also gives an upper bound on the number of vertices in the label cover graph $G$.

\paragraph{\SMP Protocol.} Observe first that the trivial protocol, one where every player sends the whole input to the referee, does not suffice for us; this is because the message length $\ell$ is $\Omega(\log n)$, meaning that the running time of the reduction is $n^{\Omega(h)} = n^{\Omega(k^2)}$ which is not FPT time.

Nevertheless, there still is a simple protocol that does the job. Notice that the input vertices $u^{\{i, j\}}_i$ and $u^{\{i, j\}}_j$ given to Player $\{i, j\}$ are already promised to form an edge. Hence, the only thing the referee needs to check is whether each alleged vertex of the clique sent to different players are the same; namely, he only needs to verify that, for every $i \in [k]$, we have $u^{\{i, 1\}}_i = u^{\{i, 2\}}_i = \cdots = u^{\{i, i - 1\}}_i = u^{\{i, i + 1\}}_i = \cdots = u^{\{i, k\}}_i$. In other words, he only needs to check equalities for each of the $k$ unknowns. The equality problem and its variants are extensively studied in communication complexity (see, e.g., \cite{Yao79,KushilevitzNisan96-CCBook}). In our case, the protocol can be easily derived using any error-correcting code. Specifically, for an outcome $\gamma \in \{0, 1\}^r$ of the random coin tosses, every Player $\{i, j\}$ encodes each of his input ($u^{\{i, j\}}_i$ and $u^{\{i, j\}}_j$) using a binary error-correcting code and sends only the $\gamma$-th bit of each encoded word to the referee. The referee then checks whether, for every $i \in [k]$, the received $\gamma$-th bits of the encodings of $u^{\{i, 1\}}_i, u^{\{i, 2\}}_i, \dots, u^{\{i, k\}}_i$ are equal. 

In the protocol described above, the message length $\ell$ is now only two bits (one bit per vertex), the randomness $r$ used is logarithmic in the block length of the code, the soundness $s$ is one minus the relative distance of the code. If we use a binary code with constant rate and constant relative distance (aka \emph{good codes}), then $r$ will be simply $O(\log \log n)$; this means that the running time of the reduction is $\poly(n, \exp(O(k^2)))$ as desired. While the soundness in this case will just be some constant less than one, we can amplify the soundness by repeating the protocol multiple times independently; this increases the randomness and message length, but it is still not hard to see that, with the right number of repetitions, all parameters lie within the desired ranges. With this, we have completed our sketch for the proof of $\W[1]$-hardness of approximating \maxcover.

\subsection{A Framework for Parameterized Hardness of Approximation}\label{subsec:framework}

%The $\W[1]$-hardness proof sketch above demonstrates the main steps in the our framework: find an associated communication problem for the hypothesis and implement a protocol for this problem. But how, in general, do we find a communication problem associated with the hypothesis? 
The $\W[1]$-hardness proof sketched above is an example of a much more general connection between communication protocol and the hardness of approximating \maxcover. To gain  insight on this, consider any function $f: X_1 \times \cdots \times X_k \to \{0, 1\}$. This function naturally induces both a communication problem and a computational problem. The communication problem for $f$ is one where there are $k$ players, each player $i$ receives an input $a_i \in X_i$, and they together wish to compute $f(a_1, \dots, a_k)$. The computational problem for $f$, which we call the \emph{Product Space Problem}\footnote{The naming comes from the product structure of the domain of $f$.} of $f$ (abbreviated as $\PSP(f)$), is one where the input consists of subsets $A_1 \subseteq X_1, \dots, A_k \subseteq X_k$ and the goal is to determine whether there exists $(a_1, \dots, a_k) \in A_1 \times \cdots \times A_k$ such that $f(a_1, \dots, a_k) = 1$. The sketch reduction to \maxcover above in fact not only applies to the specific communication problem of $k$-\clique: the analogous construction is a generic way to translate any \SMP protocol for the communication problem of any function $f$ to a reduction from $\PSP(f)$ to \maxcover. To phrase it somewhat differently, if we have an \SMP protocol for $f$ with certain parameters and $\PSP(f)$ is hard to solve, then \maxcover is hard to approximate.

This brings us to the framework we use in this paper. It consists of only two steps. First, we rewrite the problem in the hypotheses as Product Space Problems of some family of functions $\cF$. This gives us the conditional hardness for solving $\PSP(\cF)$. Second, we devise an \SMP protocol for every function $f \in \cF$. Given the connection outlined in the previous paragraph, this automatically yields the parameterized hardness of approximating \maxcover. 

To gain more intuition into the framework, note that in the case of $k$-\clique above, the function $f \in \cF$ we consider is just the function $f: X_{\{1, 2\}} \times \cdots \times X_{\{k, k - 1\}}$ where each of $X_{\{1, 2\}}, \cdots, X_{\{k, k - 1\}}$ is a copy of the edge set. The function $f$ ``checks'' that the edges selected form a clique, i.e., that, for every $i \in [k]$, the alleged $i$-th vertex of the clique specified in the $\{i, j\}$-coordinate is equal for every $j \ne i$. Since this is a generalization of the equality function, we call such a class of functions ``multi-equality''. It turns out that 3-\SAT can also be written as \PSP of multi-equality; each $X_i$ contains assignments to $1/k$ fraction of the clauses and the function $f$ checks that each variable is assigned the same value across all $X_i$'s they appear in. A protocol essentially the same as the one given above also works in this setting and immediately gives our \ETH-hardness result (Theorem~\ref{thm:domETH})! Unfortunately, this does not suffice for our \SETH-hardness. In that case, the function used is the $k$-way set disjointness; this interpretation of \SETH is well-known (see, e.g., \cite{W05}) and is also used in~\cite{ARW17}. Lastly, the \ksum problem is already written in \PSP form where $f$ is just the \sumz function that checks whether the sum of $k$ specified numbers equals to zero.

%To give even more intuition on the framework, let us note 

%Once we do so, the communication protocol naturally yields reduction to \maxcover as outlined above.

%Another way to interpret the communication problem is to first rewrite the problems in the hypotheses as what we call a \emph{Product Space Problems} (\PSP). In \PSP, 

%Before we ended this subsection and moved on to the comparisons between our work and previous works, l
Let us note that in the actual proof, we have to be more careful about the parameters than in the above sketch. Specifically, the reduction from \maxcover to $k$-\domset from~\cite{ChalermsookCKLMNT17} incurs a blow-up in size that is exponential in terms of the number of vertices in each left super-node (i.e., exponential in $|U_\gamma|$). This means that we need $|U_1|, \dots, |U_r| = o(\log n)$. In the context of communication protocol, this translates to keeping the message length $O(\log \log n)$ where $O(\cdot)$ hides a sufficiently small constant. Nevertheless, for the protocol for $k$-\clique reduction (and more generally for multi-equality), this does not pose a problem for us since the message length before repetitions is $O(1)$ bits; we can make sure that we apply only $O(\log \log n)$ repetitions to the protocol. % This is also the reason why the gap we achieve in $k$-\domset is $(\log n)^{\poly(1/k)}$.

For \sumz, known protocols either violate the above requirement on message length~\cite{Nisan94} or use too much randomness~\cite{Viola15}. Nonetheless, a simple observation allows us to compose Nisan's protocol~\cite{Nisan94} and Viola's protocol~\cite{Viola15} and arrive at a protocol with the best of both parameters. This new protocol may also be of independent interest beyond the scope of our work.

On the other hand, well-known communication complexity lower bounds on set disjointness~\cite{R92,KS92,BJKS04} rule out the existence of protocols with parameters we wish to have! \cite{ARW17} also ran into this issue; in our language, they got around this problem by allowing the referee to receive an advice. This will also be the route we take. Even with advice, however, devising a protocol with the desired parameters is a technically challenging issue. In particular, until very recently, no protocol for set disjointness with $O(\log \log n)$ message length (and $o(n)$ advice length) was known. This was overcome in the work of Rubinstein~\cite{R17} who used algebraic geometric codes to give such a protocol for the two-player case. We extend his protocol in a straightforward manner to the $k$-player case; this extension was also suggested to us by Rubinstein~\cite{Rub17com}.

A diagram illustrating the overview of our approach can be found in Figure~\ref{fig:overview}.

\begin{figure}[h!]
    \centering
    \resizebox{\textwidth}{!}{\begin{tikzpicture}
\node (wone) [draw=red!80!black,thick] at (-0.7, 0) {$\W[1] \ne \FPT$};
\node (eth) [draw=red!80!black,thick] at (-0.7, -1.5) {ETH};
\node (seth) [draw=red!80!black,thick] at (-0.7, -3) {SETH};
\node (ksum) [draw=red!80!black,thick] at (-0.7, -4.5) {$k$-Sum Hyp.};
\node (multeq) [draw=red!80!black,thick] at (5, -0.75) {$\PSP(\EQ)$};
\node (disj) [draw=red!80!black,thick] at (5, -3) {$\PSP(\SD)$};
\node (sumzero) [draw=red!80!black,thick] at (5, -4.5) {$\PSP(\SZ)$};
\node (maxcov) [draw=red!80!black,thick] at (10, -2.25) {\maxcover};
\node (domset) [draw=red!80!black,thick] at (14.5, -2.25) {$k$-\domset};
\draw [-{Latex[length=1.5mm, width=1.5mm]}] (wone) -- (multeq);
\draw [-{Latex[length=1.5mm, width=1.5mm]}] (eth) -- (multeq);
\draw [-{Latex[length=1.5mm, width=1.5mm]}] (seth) -- (disj);
\draw [-{Latex[length=2mm, width=1.5mm]}] (multeq) -- (maxcov);
\draw [-{Latex[length=1.5mm, width=1.5mm]}] (disj) -- (maxcov);
\draw [-{Latex[length=1.5mm, width=1.5mm]}] (sumzero) -- (maxcov);
\draw [-{Latex[length=1.5mm, width=1.5mm]}] (ksum) -- (sumzero); 
\draw [-{Latex[length=1.5mm, width=1.5mm]}] (maxcov) -- (domset);
\node [above, align=center] at (1.9, 1.4) {\footnotesize Rewriting Hypotheses};
\node [above, align=center] at (1.9, 1.05) {\footnotesize  in PSP form};
\node [above, align=center] at (1.9, 0.6) {\footnotesize (Section~\ref{sec:typeA})};
\node[above, align=center] at (8, 1.4) {\footnotesize Connection between SMP};
\node[above, align=center] at (8, 1.05) {\footnotesize Protocol and \maxcover};
\node[above, align=center] at (8, 0.6) {\footnotesize (Section~\ref{sec:typeA:cc2maxcover})};
\node[above, align=center] at (8.1, -1.2) {\scriptsize Protocol for \EQ};
\node[above, align=center] at (8.1, -1.5) {\scriptsize (Section~\ref{sec:multeq})};
\node[above, align=center] at (7.4, -2.4) {\scriptsize Protocol for \SD};
\node[above, align=center] at (7.4, -2.7) {\scriptsize (Section~\ref{sec:SETH})};
\node[above, align=center] at (8.4, -4) {\scriptsize Protocol for \SZ};
\node[above, align=center] at (8.4, -4.3) {\scriptsize (Section~\ref{sec:ksum})};
\node [above, align=center] at (12.25, -1.2) {\footnotesize Reduction from};
\node [above, align=center] at (12.25, -1.65) {\footnotesize \cite{ChalermsookCKLMNT17}};
\node [above, align=center] at (12.25, -2.1) {\footnotesize (Appendix~\ref{app:maxcov-domset})};
\end{tikzpicture}}
    \caption{Overview of Our Framework. The first step is to reformulate each hypothesis in terms of hardness of a \PSP problem, which is done in Section~\ref{sec:typeA}. Using the connection between SMP protocols and \maxcover outlined earlier (and formalized in Section~\ref{sec:typeA:cc2maxcover}), our task is now to devise SMP protocols with certain parameters for the corresponding communication problems; these are taken care of in Sections~\ref{sec:SETH},~\ref{sec:multeq} and~\ref{sec:ksum}. For completeness, the final reduction from \maxcover to $k$-\domset which was shown in \cite{ChalermsookCKLMNT17} is included in Appendix~\ref{app:maxcov-domset}.} \label{fig:overview}
\end{figure}
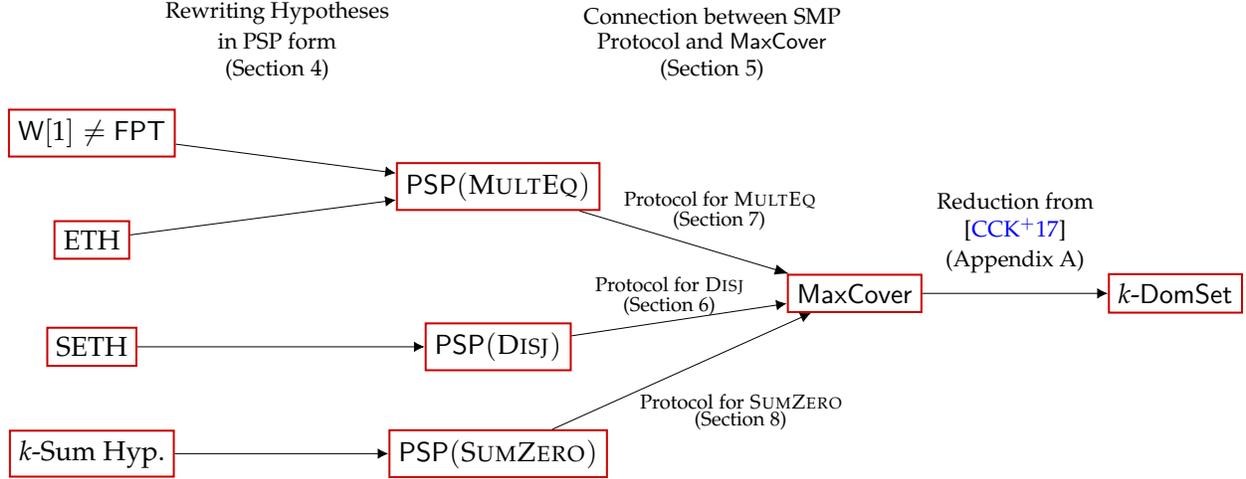

\paragraph{Comparison to Abboud et al.}
The main result of Abboud et al.\ \cite{ARW17} is their \SETH-hardness of the gap label cover problem which they refer to as the \PCP-Vectors problem. In fact, \PCP-Vectors is equivalent to \maxcover when $h=2$ (i.e., the number of right super nodes is two). However, formulating the label cover problem as \maxcover instead of \PCP-Vectors is beneficial for us, as our goal it to reduce to graph problems. 

In their work, they merge the roles of the referee and the first player as it is necessary to achieve the goal of proving hardness of approximation for important problems in \P\ (which are usually defined on one or two sets of vectors). However, by doing this the details of the proof become a little convoluted. On the contrary, our framework with the \SMP model is arguably a cleaner framework to work with and it works well for our goal of proving hardness of approximation for parameterized problems.

Finally, we note that our observation that the hardness of approximating \maxcover can be obtained from any arbitrary hypothesis as long as there is an underlying product structure (as formalized via \PSP{}s) is a new contribution of this paper.

\section{Preliminaries and Backgrounds}
\label{sec:prelim}

	We use standard graph terminology.
Let $G$ be any graph.
The vertex and edge sets of $G$ are denoted by $V(G)$ and $E(G)$, respectively.
We say that a subset of vertices $S\subseteq V(G)$ {\em dominates} $G$ 
if every vertex $v\in S\setminus V(G)$ has a neighbor in $S$,
and we call $S$ a {\em dominating set} of $G$.
A {\em $k$-dominating set} of $G$ is a dominating set of $G$ with cardinality $k$.
The {\em domination number}, denoted by (\domset(G)) of $G$ is the size of the smallest dominating set of $G$.
A {\em clique} $H$ in $G$ is a complete subgraph of $G$,
and we say that $H$ is a {\em $k$-clique} if $H$ contains $k$ vertices.
The {\em clique number} of $G$ is the size of the largest clique in $G$. 
Sometime we abuse notation and call a subset of vertices $S\subseteq V(G)$, a clique,
if $S$ induces a complete subgraph in $G$.

\subsection{Problem Definitions}

Below are the list of problems considered in this paper.

\begin{itemize}
\item {\bf Dominating Set.}
In the {\em $k$-Dominating Set} problem ($k$-\domset), 
we are given a graph $G$, and the goal is to decide whether $G$ has a dominating set of size $k$. 
In the minimization version, called {\em Minimum Dominating Set} (\domset, for short), the goal is 
to find a dominating set in $G$ of minimum size.

\item {\bf Clique.}
In the {\em $k$-Clique} problem ($k$-\clique), 
we are given a graph $G$, and the goal is to decide whether $G$ has a clique of size $k$.
In the maximization version, called {\em Maximum Clique} (\clique, for short), the goal is 
to find a clique in $G$ of maximum size.

\item {\bf $k$-SAT.}
In the $k$-SAT problem ($k$-\SAT), we are given a CNF formula $\Phi$ with \textbf{at most} $k$ literals in a clause and the goal is to decide whether $\Phi$ is satisfiable.

\item {\bf $k$-Sum.}
In \ksum, we are given $k$ subsets $S_1, \dots, S_k \subseteq ([-M, M] \cap \Z)$ of integers between $-M$ and $M$ (inclusive), and the goal is to determine whether there exist $x_1 \in S_1, \dots, x_k \in S_k$ such that $x_1 + \cdots + x_k = 0$. That is, we wish to pick one integer from each subset so that they sum to zero.

\end{itemize}

In addition to the above problems, 
we devote the next section to define and discuss a variant of the {\em label cover} problem, namely 
\maxcover.

\subsection{\maxcover\ -- A Variant of Label Cover}
\label{sec:prelim:label-cover}

We now define a variant of the label cover problem called \maxcover, which was introduced by Chalermsook~et~al.~\cite{ChalermsookCKLMNT17} to capture the parameterized inapproximability of $k$-\clique and $k$-\domset. The approximation hardness of \maxcover will be the basis of our hardness results.
%% for $k$-\domset.

The input of \maxcover is a \emph{label cover instance}; a label cover instance $\Gamma$ consists of a bipartite graph $G=(U, W; E)$ such that $U$ is partitioned into $U=U_1\cup \cdots \cup U_\ell$ and $W$ is partitioned into $W=W_1\cup \cdots \cup W_h$. We sometimes refer to $U_{i}$'s and $W_j$'s as {\em left super-nodes} and {\em right super-nodes} of $\Gamma$, respectively. Another parameter we will be interested in is the maximum size of left super nodes, i.e., $\underset{i \in [\ell]}{\max}\ |U_i|$; we refer to this quantity as the \emph{left alphabet size} of the instance.

A solution to \maxcover is called a {\em labeling},
which is a subset of vertices $S \subseteq W$ formed by picking a vertex $w_{j}$ from each $W_j$ (i.e., $|S\cap W_j| = 1$ for all $j \in [h]$). We say that a labeling $S$ {\em covers} a left super-node $U_{i}$ if there exists a vertex $u_{i} \in U_{i}$ such that $u_i$ is a neighbor of every vertex in $S$.
The goal in \maxcover is to find a labeling that covers the maximum fraction of left super-nodes. We abuse the notation \maxcover and also use it for the optimum as well, i.e.,
\begin{align*}
\maxcover(\Gamma) = \frac{1}{\ell} \left(\max_{\text{labeling } S} |\{i \in [\ell] \mid U_i \text{ is covered by } S\}|\right). 
\end{align*}

The above terminologies for \maxcover are from~\cite{ChalermsookCKLMNT17}. Note, however, that our definitions are phrased somewhat different than theirs; in our definitions, the input graphs are the so-called {\em label-extended graphs} whereas in their definitions, the input graphs are the {\em constraint graphs}. Nevertheless, it is not hard to see that the two versions are in fact equivalent. Another difference is that we use $\maxcover$ to denote the \emph{fraction} of left super-nodes covered by the optimal labeling whereas \cite{ChalermsookCKLMNT17} uses the notion for the \emph{number} of covered left super-nodes. The former notation is somewhat more convenient for us as the value is between zero and one.

\subsection{Inapproximability of \domset from \maxcover}
\label{sec:dom-set}

The relation between \maxcover and \domset has been observed in literature. 
The {\em $k$-prover system} introduced by Feige in \cite{Feige98} can be casted as 
a special case of \maxcover with {\em projection property}, and
it has been shown that this proof system can be transformed into an instance of \domset.
We note, however, that the optimal value of the \domset instance produced 
by Feige's $k$-prover system has size dependent on 
the number of left super-nodes rather than $k$, the number of right super-nodes.
Recently, Chalermsook~et~al.~\cite{ChalermsookCKLMNT17} observed that even without
the projection property, the relation between \maxcover and \domset still holds,
and the value of the optimal solution can be reduced to $k$. This is stated formally below. 

\begin{thm}[Reduction from \maxcover to \domset~\cite{ChalermsookCKLMNT17}] \label{lem:red-maxcov-domset}
There is an algorithm that, given a \maxcover instance 
$\Gamma = (U = \bigcup_{j=1}^{q} U_j, W = \bigcup_{i=1}^{k} W_i, E)$,
outputs a $k$-\domset instance 
$G$ such that
\begin{itemize}
\item If $\maxcover(\Gamma) = 1$, then $\domset(G) = k$.
\item If $\maxcover(\Gamma) \leq \varepsilon$, then $\domset(G) \geq (1/\varepsilon)^{1/k} \cdot k$.
\item $|V(G)| = |W| + \sum_{j \in [q]} k^{|U_j|}$.
\item The reduction runs in time $O\left(|W| \left(\sum_{j \in [q]} k^{|U_j|}\right)\right)$.
\end{itemize}
\end{thm}

For the sake of self-containedness, we provide the proof of Theorem~\ref{lem:red-maxcov-domset} in Appendix~\ref{app:maxcov-domset}.

\subsection{The Hypotheses}

We now list and discuss the computational complexity hypotheses on which our results are based.

\subsubsection{\Wone Hypothesis}

The first hypothesis is \Wone, which is one of the most popular hypotheses used in the area of parameterized complexity since many fundamental parameterized problems turn out to be $\W[1]$-hard. For the interest of space, we do not give the full definition of \W-hierarchy; we refer the readers to standard textbook in the field (e.g.~\cite{DowneyF13,CyganFKLMPPS15}) for the definition and discussions regarding the hierarchy. Rather, since it is well-know that $k$-\clique is $\W[1]$-complete, we will use a more convenient formulation of \Wone, which simply states that $k$-\clique is not in \FPT:

\begin{hypothesis}[\Wone Hypothesis]
For any computable function $T: \N \to \N$, no algorithm can solve $k$-\clique in $T(k)\cdot\poly(n)$ time where $n$ denotes the number of vertices in the input graph.
\end{hypothesis}

\subsubsection{Exponential Time Hypothesis and Strong Exponential Time Hypothesis}

Our second hypothesis is the Exponential Time Hypothesis (\ETH), which can be stated as follows.  

\begin{hypothesis}[Exponential Time Hypothesis (\ETH)~\cite{IP01,IPZ01,Tovey84}] \label{hyp:eth}
There exists $\delta > 0$ such that no algorithm can solve 3-\SAT in $O(2^{\delta n})$ time where $n$ is the number of variables. Moreover, this holds even when restricted to formulae in which each variable appears in at most three clauses.
\end{hypothesis}

Note that the original version of the hypothesis from~\cite{IP01} does not enforce the requirement that each variable appears in at most three clauses. To arrive at the above formulation, we first apply the Sparsification Lemma of~\cite{IPZ01}, which implies that we can assume without loss of generality that the number of clauses $m$ is $O(n)$. We then apply Tovey's reduction~\cite{Tovey84} which produces a 3-\CNF instance with at most $3m + n = O(n)$ variables and every variable occurs in at most three clauses. This means that the bounded occurrence restriction is also without loss of generality.

%The last claim on the number of clauses each variable appears in follows from the reduction described in \cite{Tovey84}. In a subsequent work to \cite{IP01}, the Strong Exponential Time Hypothesis (\SETH) was also introduced.

We will also use a stronger hypothesis called the Strong Exponential Time Hypothesis (\SETH):

\begin{hypothesis}[Strong Exponential Time Hypothesis (\SETH)~\cite{IP01,IPZ01}] \label{hyp:seth}
For every $\varepsilon > 0$, there exists $k = k(\varepsilon) \in \N$ such that no algorithm can solve $k$-\SAT in $O(2^{(1 - \varepsilon)n})$ time where $n$ is the number of variables. Moreover, this holds even when the number of clauses $m$ is at most $c(\varepsilon) \cdot n$ where $c(\varepsilon)$ denotes a constant that depends only on $\varepsilon$.
\end{hypothesis}

Again, we note that, in the original form~\cite{IP01}, the bound on the number of clauses is not enforced. However, the Sparsification Lemma~\cite{IPZ01} allows us to do so without loss of generality. 

\subsubsection{\ksum Hypothesis}

Our final hypothesis is the \ksum Hypothesis, which can be stated as follows. 

\begin{hypothesis}[\ksum Hypothesis~\cite{AL13}] \label{h:ksum} \label{hyp:ksumh}
For every integer $k \geq 3$ and every $\varepsilon > 0$, no $O(n^{\lceil k/2 \rceil - \varepsilon})$ time algorithm can solve \ksum where $n$ denotes the total number of input integers, i.e., $n = |S_1| + \cdots + |S_k|$. Moreover, this holds even when $M = n^{2k}$.
\end{hypothesis}

The above hypothesis is a natural extension of the more well-known \threesum Hypothesis~\cite{GajentaanO95,P10}, which states that \threesum cannot be solved in $O(n^{2 - \varepsilon})$ time for any $\varepsilon > 0$. Moreover, the \ksum Hypothesis is closely related to the question of whether \subsetsum can be solved in $O(2^{(1/2 - \varepsilon)n})$ time; if the answer to this question is negative, then \ksum cannot be solved in $O(n^{k/2 - \varepsilon})$ time for every $\varepsilon > 0, k \in \mathbb{N}$. We remark that, if one is only willing to assume this latter weaker lower bound of $O(n^{k/2 - \varepsilon})$ instead of $O(n^{\lceil k/2\rceil - \varepsilon})$, our reduction would give an $O(n^{k/2 - \varepsilon})$ running time lower bound for approximating $k$-\domset. Finally, we note that the assumption that $M = n^{2k}$ can be made without loss of generality since there is a randomized reduction from the general version of the problem (where $M$ is, say, $2^n$) to this version of the problem and this reduction can be derandomized under a certain circuit complexity assumption~\cite{AbboudLW14}.

\subsection{Error-Correcting Codes}

An error correcting code $C$ over alphabet $\Sigma$ is a function $C: \Sigma^m \to \Sigma^d$ where $m$ and $d$ are positive integers which are referred to as the {\em message length} and {\em block length} of $C$ respectively. Intuitively, the function $C$ encodes an original message of length $m$ to an encoded message of length $d$.
Since we will also deal with communication protocols, for which ``message length'' has another meaning, we will sometimes refer to the message length of codes as \emph{code message length} whenever there is an ambiguity.
The {\em rate} of a code $\rho(C)$ is defined as the ratio between its message length and its block length, i.e., $\rho(C) = \nicefrac{m}{d}$. 
The {\em relative distance} of a code, denoted by $\delta(C)$, is defined as $\underset{x \ne y \in \Sigma^m}{\min} \delta(C(x), C(y))$ where $\delta(C(x), C(y))$ is the {\em relative Hamming distance} between $C(x)$ and $C(y)$, i.e., the fraction of coordinates on which $C(x)$ and $C(y)$ disagree.
%
%
%When $\Sigma$ forms a group structure under the additive operator, a code $C$ on alphabet set $\Sigma$ is said to be linear if $C(x) + C(y) = C(x + y)$ where the vector additions are done coordinate-wise.

\subsubsection{Good Codes}

In the construction of our communication protocol in Section~\ref{sec:multeq}, 
we require our codes to have constant rate and constant relative distance (referred to as \emph{good codes}).
%More precisely, given a message of length $n$, we need a codes whose code-message length $O(\log_{|\Sigma|}n)$, where $\Sigma$ is the alphabet set, and with relative distance 
%$1/|\Sigma|-\epsilon$, for some small constant $\epsilon>0$. 
%
%The codes that satisfies this property is the {\em random codes} constructed by taking a random
%subset of $\Sigma^{\log_{|\Sigma|}n}$. 
It is not hard to see that random codes, ones where each codeword $C(x)$ is randomly selected from $\Sigma^d$ independently from each other, satisfy these properties.
For binary codes (i.e., $|\Sigma| = 2$), one can explicitly construct such codes using 
expander graphs (so called {\em Expander Codes}~\cite{SipserS96}); alternatively {\em Justesen Code}~\cite{Justesen72} also have the same property (see Appendix E.1.2.5 from \cite{G09} for an excellent exposition). 

\begin{fact} \label{lem:good-code}
For some absolute constant $\delta, \rho > 0$, 
there exists a family of codes $\mathcal{C}:=\{C_m:\{0,1\}^m\to \{0,1\}^{d(m)}\}_{m\in\N}$  such that for every $m\in\N$ the rate of $C_m$ is at least $\rho$ and the relative distance of $C_m$ is at least $\delta$. Moreover, any codeword of $C_m$ can be computed in time $\poly(m)$.
\end{fact}

\subsubsection{Algebraic Geometric Codes}
In the construction of our communication protocol in Section~\ref{sec:SETH}, 
we require our codes to have some special algebraic properties which have been shown to be present in algebraic geometric codes \cite{GS96}. First, we will introduce a couple of additional definitions. 

\begin{definition}[Systematicity]
Given $s \in \mathbb N$, a code $C:\Sigma^m\to \Sigma^{d}$ is {\em $s$-systematic} if there exists a size-$s$ subset of $[d]$, which for convenience we identify with $[s]$, such that for every $x \in \Sigma^{s}$ there exists $w \in \Sigma^m$ in which $x = C(w)\mid_{[s]}$. 
\end{definition}

\begin{definition}[Degree-$t$ Closure] \label{def:poly}
Let $\Sigma$ be a finite field. 
Given two codes $C:\Sigma^m\to \Sigma^{d},C':\Sigma^{m'}\to \Sigma^{d}$ and positive integer $t$, we say that $C'$ is a degree-$t$ closure of $C$ if, for every $w_1,\ldots ,w_r\in \Sigma^m$ and $P\in \mathbb F[X_1,\ldots ,X_r]$ of total degree at most $t$, it holds that $\omega:=P(C(w_1),\ldots ,C(w_r))$ is in the range of $C'$, where $\omega\in \Sigma^d$ is defined coordinate-wise by the equation $\omega_i:=P(C(w_1)_i,\ldots ,C(w_r)_i)$.
\end{definition}

Below we provide a self-contained statement of the result we rely on in Section~\ref{sec:SETH}; it follows from Theorem~7 of \cite{SAKSD01}, which
gives an efficient construction of the algebraic geometric codes based on \cite{GS96}'s explicit towers of function fields.

\begin{theorem}[\cite{GS96,SAKSD01}] \label{thm:ag-code}
There are two polynomial functions $\hat{r},\hat{q}:\mathbb{N}\to\mathbb{N}$ such that for every $k \in \mathbb{N}$ and any prime $q>\hat{q}(k)$, there are two code families $\mathcal{A} = \{A_n\}_{n \in \mathbb{N}},\ \mathcal{B} = \{B_n\}_{n \in \mathbb{N}}$ such that the following holds for all $n\in\mathbb N$,
\begin{itemize}
\item  $A_n$ and $B_n$ are $n$-systematic code with alphabet $\mathbb{F}_{q^2}$,
\item  $A_n$ and $B_n$ have block length less than $n\cdot \hat{r}(k)$.
\item  $B_n$ has relative distance $ \geq \nicefrac{1}{2}$,
\item  $B_n$ is a degree-$k$ closure of $A_n$, and,
\item  Any codeword in $A_n$ or $B_n$ can be computed in poly($n$) time .
\end{itemize}
\end{theorem}

We remark here that variants of the above theorem have previously found applications in the construction of special kinds of PCPs \cite{BCGRS17,BKKMS16}. In these works, the theorems are also stated in a language similar to Theorem~\ref{thm:ag-code} above.

\section{Product Space Problems and Popular Hypotheses}
\label{sec:typeA}

In this section, we define a class of computational problems called Product Space Problems (\PSP). As the name suggests, a problem in this class is defined on a class of functions whose domain is a $k$-ary Cartesian Product, i.e., $f: X_1 \times \cdots \times X_k \to \{0, 1\}$. The input of the problem are subsets\footnote{Each $A_i$ will be explicitly given as part of the input through the elements that it contains.} $A_1 \subseteq X_1, \dots, A_k \subseteq X_k$, and the goal is to determine whether there exists $(a_1, \dots, a_k) \in A_1 \times \cdots \times A_k$ such that $f(a_1, \dots, a_k) = 1$. The size of the problem is determined by $\underset{i \in [k]}{\max}\ |A_i|$. A formal definition of \PSP can be found below.

\begin{definition}[{\bf Product Space Problem}]
Let $m: \N \times \N \to \N$ be any function and $\mathcal{F}:=\{f_{N, k}:\{0,1\}^{m(N, k)\times k}\to\{0,1\}\}_{N,k \in \N}$ be a family of Boolean functions indexed by $N$ and $k$. For each $k \in \N$, the \emph{product space problem} \PSP{}$(k,\mathcal{F})$ of order $N$ is defined as follows: given $k$ subsets $A_1,\ldots ,A_k$ of $\{0,1\}^{m(N, k)}$ each of cardinality at most $N$ as input, determine if there exists $(a_1,\ldots ,a_k)\in A_1\times \cdots \times A_k$ such that $f_{N, k}(a_1,\ldots ,a_k)=1$. We use the following shorthand $\PSP(k,\mathcal{F},N)$ to describe $\PSP(k,\mathcal{F})$ of order $N$.
\end{definition}

In all the \PSP{}s considered in this paper, the input length $m(N, k)$ is always at most $\poly(k)\cdot \log N$ and $f_{N, k}$ is always computable in time $\poly(m(N, k))$. In such a case, there is a trivial $N^{k + o_k(1)}$-time algorithm to solve $\PSP(k, \cF, N)$: enumerating all $(a_1, \dots, a_k) \in A_1 \times \cdots \times A_k$ and check whether $f_{N, k}(a_1, \dots, a_k) = 1$. The rest of this section is devoted to rephrasing the hypotheses (\SETH, \ETH, \Wone and the \ksum Hypothesis) in terms of lower bounds for \PSP{}s. The function families $\cF$'s, and running time lower bounds will depend on the hypotheses. For example, \SETH will corresponds to set disjointness whereas \Wone will correspond to a generalization of equality called ``multi-equality''; the former will give an $N^{k(1 - o(1))}$ running time lower bound whereas the latter only rules out \FPT time algorithms.

We would like to remark that the class of problems called `locally-characterizable sets' introduced by Goldreich and Rothblum \cite{GR18} are closely related to \PSP{}s. Elaborating,  we may interpret locally-characterizable sets as the negation of \PSP{}s, i.e., for any $\PSP(k,\mathcal{F},N)$, we may define the corresponding locally-characterizable set $\mathcal{S}$ as follows:
$$\mathcal{S}=\{(A_1,\ldots ,A_k)\mid \text{for all }(a_1,\ldots ,a_k)\in A_1\times \cdots \times A_k \text{ we have } f_{N, k}(a_1,\ldots ,a_k)=0\}.$$  

\begin{sloppypar}Finally, we note that the class of problems called `counting local patterns' introduced in
\cite{GR18a}
are the counting counterpart of \PSP{}s, i.e., for any instance $(A_1,\ldots ,A_k)$ of $\PSP(k,\mathcal{F},N)$, we may define the corresponding counting local pattern solution to be the number of distinct $(a_1,\ldots,a_k)\in A_1\times~\cdots~\times~A_k$ such that $f_{N,k}(a_1,\ldots ,a_k)=1$.
\end{sloppypar}
\subsection{\ksum Hypothesis}

To familiarize the readers with our notations, we will start with the \ksum Hypothesis, which is readily in the \PSP form. Namely, the functions in the family are the \sumz functions that checks if the sum of $k$ integers is  zero:

\begin{definition}[\sumz]
Let $k, m \in \N$. $\SZ_{m, k}: (\{0, 1\}^m)^k \rightarrow \{0, 1\}$ is defined by
\begin{align*}
\SZ_{m, k}(x_1, \dots, x_k) = \begin{cases}
1\text{ if } \underset{i\in [k]}{\sum}x_i=0,\\
0\text{ otherwise},
\end{cases}
\end{align*}
where we think of each $x_i$ as a number in $[-2^{m-1},2^{m-1}-1]$, and the addition is over $\Z$.
\end{definition}

The function family $\cF^{\SZ}$ can now be defined as follows.

\begin{definition}[Sum-Zero Function Family]\label{def:ksumFam}
Let $m: \N \times \N \to \N$ be a function defined by $m(N, k) = 2k\lceil \log N\rceil$. $\cF^{\SZ}$ is defined as $\{\SZ_{m(N, k), k}\}_{N \in \N, k \in N}$.
\end{definition}

The following proposition is immediate from the definition of the \ksum Hypothesis.

\begin{proposition}\label{prop:ksumPSP}
Assuming the \ksum Hypothesis, for every integer $k \geq 3$ and every $\varepsilon > 0$, no $O(N^{\lceil k/2 \rceil - \varepsilon})$-time algorithm can solve $\PSP(k, \cF^{\SZ}, N)$ for all $N \in \N$.
\end{proposition}

\subsection{Set Disjointness and \SETH}
We recall the $k$-way disjointness function,
which has been studied extensively in literature 
(see, e.g., \cite{LS09} and references therein). 

\begin{defn}[Set Disjointness]
Let $k, m \in \N$. $\SD_{m, k}: (\{0, 1\}^m)^k \rightarrow \{0, 1\}$ is defined by
\begin{align*}
\SD_{m, k}(x_1, \dots, x_k) = \neg \left(\bigvee_{i \in [m]} \left(\bigwedge_{j \in [k]} (x_j)_i\right)\right).
\end{align*}
\end{defn}

The function family $\cF^{\SD}_c$ can now be defined as follows.

\begin{definition}[Set Disjointness Function Family]\label{def:SETHFam}
For every $c\in\mathbb N$, let $m_c: \N \times \N \to \N$ be a function defined by $m_c(N, k) = c\lceil k\log N\rceil$. $\cF^{\SD}_c$ is defined as $\{\SD_{m_c(N, k), k}\}_{N \in \N, k \in N}$.
\end{definition}

We have the following proposition which follows easily from the definition of \SETH and its well-known connection to the Orthogonal Vectors Hypothesis \cite{W05}.

\begin{proposition}\label{prop:SETHPSP}
Let $k\in\mathbb N$ such that $k>1$.
Assuming \SETH, for every $\varepsilon>0$ there exists $c:=c_\varepsilon\in \mathbb{N}$ such that no $O(N^{ k (1- \varepsilon)})$-time algorithm can solve $\PSP(k, \cF^{\SD}_{c}, N)$ for all $N \in \N$.
\end{proposition}
\begin{proof}
Fix $\varepsilon>0$ and $k>1$. By \SETH, there exists $w:=w(\varepsilon)\in\mathbb N$ and $c:=c(\varepsilon)\in\mathbb N$ such that no algorithm can solve $w$-\SAT in $O(2^{(1 - \varepsilon)n})$ time where $n$ is the number of variables and $m\le cn$ is the number of clauses. 
For every $w$-\SAT formula $\phi$, we will build $A_1^{\phi},\ldots ,A_{k}^{\phi}\subseteq \{0,1\}^m$ each of cardinality $N:=2^{\nicefrac{n}{k}}$ such that there exists $(a_1,\ldots ,a_k)\in A_1^{\phi}\times \cdots \times A_{k}^{\phi}$ such that $\SD_{m,k}(a_1,\ldots ,a_k)=1$ if and only if $\phi$ is satisfiable. Thus, if there was an $O(N^{ k (1- \varepsilon)})$-time algorithm that can solve $\PSP(k, \cF^{\SD}_{c}, N)$ for all $N \in \N$, then it would violate \SETH.

All that remains is to show the construction of $A_1^{\phi},\ldots ,A_{k}^{\phi}$ from $\phi$. Fix $i\in [k]$. For every partial assignment $\sigma$ to the variables $x_{(i-1)*(n/k)+1},\ldots ,x_{i*(n/k)}$ we build an $m$-bit vector $a_\sigma\in A_i^{\phi}$ as follows: $\forall j\in [m]$, we have $a_{\sigma}(j)=0$ is $\sigma$ satisfies the $j^{\text{th}}$ clause, and $a_{\sigma}(j)=1$ otherwise (i.e., the clause is not satisfied, or its satisfiability is indeterminate). It is easy to verify that there exists $(a_1,\ldots ,a_k)\in A_1^{\phi}\times \cdots \times A_{k}^{\phi}$ such that $\SD_{m,k}(a_1,\ldots ,a_k)=1$ if and only if $\phi$ is satisfiable. 
\end{proof}

We remark that we can prove a similar statement as that of Proposition~\ref{prop:SETHPSP} for \ETH: assuming \ETH, there exists $k_0$ such that for every $k>k_0$ there exists $c:=c_{k_0}\in \mathbb{N}$ such that no $O(N^{ o(k)})$-time algorithm can solve $\PSP(k, \cF^{\SD}_{c}, N)$ for all $N \in \N$. However, instead of associating \ETH with \SD, we will associate with the Boolean function \tEQ\ (which will be defined in the next subsection) and its corresponding \PSP. This is because, associating \ETH with \tEQ\  provides a more elementary proof of Theorem~\ref{thm:domETH} (in particular we will not need to use algebraic geometric codes -- which are essentially inevitable if we associate \ETH with \SD).

\subsection{\Wone Hypothesis and \ETH}

Again, we recall the $k$-way \eq function which has been studied extensively in literature (see, e.g., \cite{AMS12,ABC09,CRR14,CMY08,LV11,PVZ12} and references therein).  

\begin{definition}[\eq]
Let $k, m \in \N$. $\EQU_{m, k}: (\{0, 1\}^m)^k \rightarrow \{0, 1\}$ is defined by
\begin{align*}
\EQU_{m, k}(x_1, \dots, x_k) = \bigwedge_{i, j \in [k]} \left(x_i = x_j\right)
\end{align*}
where $x_i = x_j$ is a shorthand for $\underset{p \in [m]}{\bigwedge} (x_i)_p = (x_j)_p$.
\end{definition}

Unfortunately, the \PSP associated with \EQU~is in fact not hard: given sets $A_1, \dots, A_k$, it is easy to find whether they share an element by just sorting the combined list of $A_1 \cup \cdots \cup A_k$. Hence, we will need a generalization of the equality function to state our hard problem. Before we do so, let us first state an intermediate helper function, which is a variant of the usual equality function where some of the $k$ inputs may be designed as ``null'' and the function only checks the equality over the non-null inputs. We call this function the \seleq (\EQUN) function. For notational convenience, in the definition below, each of the $k$ inputs is now viewed as $(x_i, y_i) \in \{0, 1\}^{m - 1} \times \{\bot, \top\}$; if $y_i = \bot$, then $(x_i, y_i)$ represents the ``null'' input.

\begin{definition}[\seleq]
Let $k, m \in \N$. $\EQUN_{m, k}: (\{0, 1\}^{m - 1} \times \{\bot, \top\})^k \rightarrow \{0, 1\}$ is defined by
\begin{align*}
\EQUN_{m, k}((x_1, y_1), \dots, (x_k, y_k)) = \bigwedge_{i, j \in [k]} \left((y_i = \bot) \vee (y_j = \bot) \vee (x_i = x_j)\right).
\end{align*}
\end{definition}

Next, we introduce the variant of $\EQU$~whose associated \PSP is hard under \Wone and \ETH. In the settings of both \eq and \seleq defined above, there is only one unknown that is given in each of the $k$ inputs $a_1 \in A_1, \dots, a_k \in A_k$ and the functions check whether they are equal. The following function, which we name \multeq, is the $t$-unknown version of \seleq. Specifically, the $i^{\text{th}}$ part of the input is now a tuple $((x_{i, 1}, y_{i, 1}), \dots, (x_{i, t}, y_{i, t}))$ where $x_{i, 1}, \dots, x_{i, t}$ are bit strings representing the supposed values of the $t$ unknowns while, similar to \seleq, each $y_{i, q} \in \{\bot, \top\}$ is a symbol indicating whether $(x_{i, q}, y_{i, q})$ is the ``null'' input. Below is the formal definition of \EQ; note that for convenience, we use $(x_{i, q}, y_{i, q})_{q \in [t]}$ as a shorthand for $((x_{i, 1}, y_{i, 1}), \dots, (x_{i, t}, y_{i, t}))$, i.e., the $i^{\text{th}}$ part of the input.

\begin{definition}[\multeq]
Let $k, t \in \N$ and let $m \in \N$ be any positive integer such that $m$ is divisible by $t$. Let $m' = m/t$. $\tEQ_{m, k,t}: ((\{0, 1\}^{m' - 1} \times \{\bot, \top\})^{t})^k \rightarrow \{0, 1\}$ is defined by
\begin{align*}
\tEQ_{m, k,t}((x_{1, q}, y_{1, q})_{q \in [t]}, \dots, (x_{k, q}, y_{k, q})_{q \in [t]})
 &= \bigwedge_{q \in [t]} \EQUN_{m', k}((x_{1, q}, y_{1, q}), \dots, (x_{k, q}, y_{k, q})).
\end{align*}
\end{definition}

Next, we define the family $\cF^{\EQ}$; note that in the definition below, we simply choose $t(k)$, the number of unknowns, to be $k + \binom{k}{2} + \binom{k}{3}$. As we will see later, this is needed for \ETH-hardness. For $\W[1]$-hardness, it suffices to use a smaller number of variables. However, we choose to define $t(k)$ in such a way so that we can conveniently use one family for both \ETH and $\W[1]$-hardness.

\begin{definition} \label{def:fam-eq}
Let $t: \N \to \N$ be defined by $t(k) = k + \binom{k}{2} + \binom{k}{3}$. Let $m: \N \times \N \rightarrow \N$ be defined by $m(N, k) = t(k)\left(1 + k\lceil\log N\rceil\right)$. We define $\cF^{\EQ}$ as $\{\EQ_{m(N, k), k,t(k)}\}_{N \in \N, k \in \N}$.
\end{definition}

We next show a reduction from $k$-\clique to $\PSP(k', \cF^{\EQ})$ where $k' = \binom{k}{2}$. The overall idea of the reduction is simple. First, we associate the integers in $[k']$ naturally with the elements of $\binom{[k]}{2}$. We then create the sets $\left(A_{\{i, j\}}\right)_{\{i, j\} \subseteq [k],i\neq j}$ in such a way that each element of the set $A_{\{i, j\}}$ corresponds to picking an edge  between the $i$-th and the $j$-th vertices in the supposed $k$-clique. Then, \EQ~is used to check that these edges are consistent, i.e., that, for every $i \in [k]$, $a_{\{i, j\}}$ and $a_{\{i, j'\}}$ pick the same vertex to be the $i^{\text{th}}$ vertex in the clique for all $j, j' \in [k] \setminus \{i\}$. This idea is formalized in the following proposition and its proof.

\begin{proposition} \label{prop:w1-to-multeq}
Let $k \in \N$ and $k' = \binom{k}{2}$. There exists a $\poly(N,k)$-time reduction from any instance $(G, k)$ of \clique to an instance $(A_1, \dots, A_{k'})$ of the $\PSP(k', \cF^{\EQ}, N')$ where $N$ denotes the number of vertices of $G$ and $N' = \binom{N}{2}$.
\end{proposition}

\begin{proof}
Given a \clique instance\footnote{We assume without loss of generality that $G$ does not contain any self-loop.} $(G, k)$, the reduction proceeds as follows. For convenience, we assume that the vertex set $V(G)$ is $[N]$. Furthermore, we associate the elements of $[k']$ naturally with the elements of $\binom{[k]}{2}$. 
For the sake of conciseness, we sometimes abuse notation and think of $\{i,j\}$ as an ordered pair $(i,j)$ where $i<j$.
For every $\{i, j\} \in \binom{[k]}{2}$ such that $i<j$, the set $A_{\{i, j\}}$ contains one element $a^{\{u, v\}}_{\{i, j\}} = \left(a^{\{u, v\}}_{\{i, j\}, 1}, \dots, a^{\{u, v\}}_{\{i, j\}, t(k')}\right)$ for each edge $\{u, v\} \in E(G)$  such that $u < v$, where
\begin{align*}
a^{\{u, v\}}_{\{i, j\}, q} &=
\begin{cases}
(u , \top) & \text{ if } q = i, \\
(v, \top) & \text{ if } q = j, \\
(0^{}, \bot) & \text{ otherwise.}
\end{cases}
\end{align*}
Note that in the definition above, we view $u, v$ and $0$ as $\left(\nicefrac{m(N', k')}{t(k')}-1\right)$-bit strings, where $m: \N \times \N \to \N$ is as in Definition~\ref{def:fam-eq}. Also note that each set $A_{\{i, j\}}$ has size at most $\binom{N}{2} = N'$, meaning that $(A_{\{i, j\}})_{\{i, j\} \subseteq [k]}$ is indeed a valid instance of $\PSP(k', \cF^{\EQ}, N')$. For brevity, below we will use $f$ as a shorthand for $\EQ_{m(N', k'), k',t(k')}$.

($\Rightarrow$) Suppose that $(G, k)$ is a YES instance for \clique, i.e., there exists a $k$-clique $\{u_1, \dots, u_k\}$ in $G$. Assume without loss of generality that $u_1 < \cdots < u_k$.  We claim that, 
$$f\left(\left(a_{\{i,j\}}^{\{u_i,u_j\}}\right)_{i,j\in [k],i< j}\right) = 1.$$ 
To see that this is the case, observe that for every $q \in [t(k')]$ and for every $\{i, j\} \subseteq [k]$ such that $i<j$, we have either $a_{\{i, j\}, q}^{\{u_i,u_j\}} = (0, \bot)$ or $a_{\{i, j\}, q}^{\{u_i,u_j\}} = (u_q, \top)$. This means that, $\EQUN\left(\left(a_{\{i,j\},q}^{\{u_i,u_j\}}\right)_{i,j\in [k],i< j	}\right) = 1$ for every $q\in [t(k')]$. 

($\Leftarrow$) Suppose that $(A_{\{i, j\}})_{\{i, j\} \subseteq [k]}$ is a YES instance for $\PSP(k', \cF^{\EQ}, N')$, i.e., there exists $a^*_{\{i, j\}} \in A_{\{i, j\}}$ for every $\{i, j\} \subseteq [k]$ such that $f((a^*_{\{i, j\}})_{\{i, j\} \subseteq [k]}) = 1$. Suppose that $a^*_{\{i, j\}} = (x^*_{\{i, j\}, 1}, y^*_{\{i, j\}, 1}, \dots, x^*_{\{i, j\}, t(k')}, y^*_{\{i, j\}, t(k')})$.
From this solution $\{a^*_{\{i, j\}}\}_{\{i, j\} \subseteq [k]}$, we can recover the $k$-clique as follows. For each $i \in [k]$, pick an arbitrary $j(i) \in [k]$ that is not equal to $i$. Let $u_i$ be $x^*_{\{i, j\}, i}$. We claim that $u_1, \dots, u_k$ forms a $k$-clique in $G$. To show this, it suffices to argue that, for every distinct $i, i' \in [k]$, there is an edge between $u_i$ and $u_{i'}$ in $G$. To see that this holds, consider $a^*_{\{i, i'\}}$. Since $y^*_{\{i, i'\}, i} = y^*_{\{i, j(i)\}, i} = \top$, we have $x^*_{\{i, i'\}, i} = x^*_{\{i, j(i)\}, i} = u_i$. Similarly, we have $x^*_{\{i, i'\}, i'} = u_{i'}$. Since $a^*_{\{i, i'\}} \in A_{\{i, i'\}}$ and from how the set $A_{\{i, i'\}}$ is defined, we have $\{u_i, u_{i'}\} \in E(G)$, which concludes our proof.
\end{proof}

\begin{lem} \label{lem:w1-psp}
Assuming \Wone, for any computable function $T: \N \to \N$, there is no $T(k) \cdot \poly(N)$ time algorithm that can solve $\PSP(k, \cF^{\EQ}, N)$ for every $N, k \in \N$.
\end{lem}

\begin{proof}
Suppose for the sake of contradiction that, for some computable function $T: \N \to \N$, there is a $T(k) \cdot \poly(N)$ time algorithm $\cA$ that can solve $\PSP(k, \cF^{\EQ}, N)$ for every $N, k \in \N$. We will show that this algorithm can also be used to solve $k$-\clique parameterized by $k$ in \FPT time.

Given an instance $(G, k)$ of $k$-\clique, we first run the reduction from Proposition~\ref{prop:w1-to-multeq} to produce an instance $(A_1, \dots, A_{k'})$ of $\PSP(k', \cF^{\EQ}, N')$ in $\poly(N, k)$ time where $N = |V(G)|, N' = \binom{N}{2}$ and $k' = \binom{k}{2}$. We then run $\cA$ on $(A_1, \dots, A_{k'})$, which takes time $T(k') \cdot \poly(N')$. This means that we can also solve our $k$-\clique instance $(G, k)$ in time $\poly(N, k) + T(k') \cdot \poly(N') = \poly(N, k) + T\left(\binom{k}{2}\right) \cdot \poly(N)$, which is \FPT time. Since $k$-\clique is $\W[1]$-complete, this contradicts with \Wone.
\end{proof}

Next, we will prove \ETH-hardness of $\PSP(k, \cF^{\EQ})$. Specifically, we will reduce a 3-\SAT instance $\phi$ where each variable appears in at most three clauses to an instance of $\PSP(k, \cF^{\EQ}, N)$ where $N = 2^{O(n/k)}$ and $n$ denotes the number of variables in $\phi$. The overall idea is to partition the set of clauses into $k$ parts of equal size and use each element in $A_j$ to represent a partial assignment that satisfies all the clauses in the $j^{\text{th}}$ partition. This indeed means that each group has size $2^{O(n/k)}$ as intended. However, choosing the unknowns are not as straightforward as in the reduction from $k$-\clique above; in particular, if we view each variable by itself as an unknown, then we would have $n$ unknowns, which is much more than the designated $t(k) = k + \binom{k}{2} + \binom{k}{3}$ unknowns! This is where we use the fact that each variable appears in at most three clauses: we group the variables of $\phi$ together based on which partitions they appear in and view each group as a single variable. Since each variable appears in at most three clauses, the number of ways they can appear in the $k$ partitions is $k + \binom{k}2{} + \binom{k}{3}$ which is indeed equal to $t(k)$. The ideas are formalized below.

\begin{proposition}\label{prop:ETHPSP}
Let $k \in \N$. There exists a $\poly(N, k)$-time reduction from any instance $\phi$ of 3-\SAT such that each variable appears in at most three clauses in an instance $(A_1, \dots, A_k)$ of the $\PSP{}(k, \cF^{\EQ}, N)$ where $N = 2^{3\lceil m/k \rceil}$ and $m$ denotes the number of clauses in $\phi$.
\end{proposition}

\begin{proof}
Given a 3-\SAT formula $\phi$ such that each variable appears in at most three clauses. Let the variable set of $\phi$ be $\cZ = \{z_1, \dots, z_n\}$ and the clauses of $\phi$ be $\cC = \{C_1, \dots, C_m\}$. Then for every $k\in\mathbb N$, we produce an instance $(A_1, \dots, A_k)$ of $\PSP(k, \cF^{\EQ}, N)$ where $N=2^{3\lceil m/k\rceil}$ as follows.

First, we partition the clause set $\cC$ into $k$ parts $\cC_1, \dots, \cC_k$ each of size at most $\lceil m/k \rceil$. For each variable $z_i$, let $S_i$ denote $\{j \in [k] \mid \exists C_h \in \mathcal{C}_j \text{ such that }z_i \in C_h \text{ or }\overline{z}_i\in C_h\}$. Since every $z_i$ appears in at most three clauses, we have $S_i \in \binom{[k]}{\leq 3}$. For each $S \in \binom{[k]}{\leq 3}$, let $\nu(S)$ denote the set of all variables $z_i$'s such that $S_i = S$ (i.e. $S$ is exactly equal to the set of all partitions that $z_i$ appears in). The general idea of the reduction is that we will view a partial assignment to the variables in $\nu(S)$ as an unknown for $\EQ$; let us call this unknown $X_{S}$ (hence there are $k + \binom{k}{2} + \binom{k}{3} = t(k)$ unknowns). For each $j \in [k]$, $A_j$ contains one element for each partial assignment to the variables that appear in the clauses in $\cC_j$  and that satisfies all the clauses in $\cC_j$. Such a partial assignment specifies $\left(1 + k + \binom{k}{2}\right)$ unknowns: all the $X_S$ such that $j \in S$. The \EQ~function is then used to check the consistency between the partial assignments to the variables from different $A_j$'s.

To formalize this intuition, we first define more notations. Let $\widetilde{m}=3k\lceil m/k\rceil$. For every subsets $T \subseteq T' \subseteq \cZ$ and every partial assignment $\alpha: T' \to \{0, 1\}$, the restriction of $\alpha$ to $T$, denoted by $\alpha|_T$ is the function from $T$ to $\{0, 1\}$ where $\alpha|_T(z) = \alpha(z)$ for every $z \in T$. Furthermore, we define the operator $\ext(\alpha)$, which ``extends'' $\alpha$ to $\widetilde m$ bits, i.e., the $i$-th bit of $\ext(\alpha)$ is $\alpha(z_i)$ if $z_i \in T$ and is zero otherwise. Finally, we use $\var(\cC_j)$ to denote the set of all variables that appear in at least one of the clauses from $\cC_j$, i.e., $\var(\cC_j) =  \bigcup_{C \in \cC_j} \var(C)$ where $\var(C)$ denotes $\{z_i \in \cZ \mid z_i \in C \text{ or }\overline{z}_i\in C\}$.

Now, since our $t(k)$ is exactly $\left|\binom{[k]}{\le 3}\right|$, we can associate each element of $[t]$ with a subset $S \in \binom{[k]}{\leq 3}$. 
Specifically, for each partial assignment $\alpha: \var(\cC_j) \rightarrow \{0, 1\}$ such that $\alpha$ satisfies all the clauses in $\cC_j$, the set $A_j$ contains an element $a_j^\alpha = (a^\alpha_{j, S})_{S \in \binom{[k]}{\leq 3}}$ where, for every $S \in \binom{[k]}{\leq 3}$,
\begin{align*}
a^\alpha_{j, S} =
\begin{cases}
\left(\ext\left(\alpha|_{\nu(S)}\right), \top\right) & \text{if } j \in S, \\
(0^{\widetilde m}, \bot) & \text{otherwise.}
\end{cases}
\end{align*}
Fix $S \in \binom{[k]}{\leq 3}$. For every $j\in S$, observe that $\nu(S) \subseteq \var(\cC_j)$. Moreover, since each $\cC_j$ contains at most $\lceil m/k \rceil$ clauses, there are at most $3\lceil m/k \rceil$ variables in $\var(\cC_j)$. This means that $A_j$ has size at most $2^{3\lceil m/k \rceil}$. Hence, $(A_1, \dots, A_k)$ is indeed a valid instance of \PSP($k, \cF^{\EQ}, N)$ where $N = 2^{3\lceil m/k \rceil}$. For brevity, below we will use $f$ as a shorthand for $\EQ_{m(N, k), k,t(k)}$.

($\Rightarrow$) Suppose that $\phi$ is satisfiable. Let $\alpha: \cC \to \{0, 1\}$ be an assignment that satisfies all the clauses. Let $a^*_j = a^{\alpha|_{\var(\cC_j)}}_j \in A_j$ for every $j \in [k]$. Observe that, for every $S \in \binom{[k]}{\leq 3}$ and every $j \in [k]$, we either have $a^*_{j, S} = (0^{\widetilde m}, \bot)$ or $a^*_{j, S} = (\ext(\alpha|_{\nu(S)}), \top)$. This indeed implies that $f(a^*_1, \dots, a^*_k) = 1$.

($\Leftarrow$) Suppose that there exists $(a^{\alpha_1}_1, \dots, a^{\alpha_k}_k) \in A_1 \times \cdots \times A_k$ such that $f(a^{\alpha_1}_1, \dots, a^{\alpha_k}_k) = 1$. We construct an assignment $\alpha: \cZ \to \{0, 1\}$ as follows. For each $i \in [n]$, pick an arbitrary $j(i)\in[k]$ such that $z_i \in \var(\cC_{j(i)})$ and let $\alpha(z_i) = \alpha_{j(i)}(z_i)$. We claim that $\alpha$ satisfies every clause. To see this, consider any clause $C \in \cC$. Suppose that $C$ is in the partition $\cC_{j}$. It is easy to check that $f(a^{\alpha_1}_1, \dots, a^{\alpha_k}_k) = 1$ implies that $\alpha|_{\var(C)} = \alpha_j|_{\var(C)}$. Since $\alpha_j$ is a partial assignment that satisfies $C$, $\alpha$ must also satisfy $C$. In other words, $\alpha$ satisfies all clauses of $\phi$.
\end{proof}

\begin{lem} \label{lem:eth-psp}
Assuming \ETH, for any computable function $T: \N \to \N$, there is no $T(k) \cdot N^{o(k)}$ time algorithm that can solve $\PSP(k, \cF^{\EQ}, N)$ for every $N,k \in \N$.
\end{lem}

\begin{proof}
Let $\delta > 0$ be the constant in the running time lower bound in \ETH. 
Suppose for the sake of contradiction that \ETH holds but, for some function $T$, there is a $T(k) \cdot N^{o(k)}$ time algorithm $\cA$ that can solve $\PSP(k, \cF^{\EQ}, N)$ for every $N,k \in \N$. Thus, there exists a sufficiently large $k$ such that the running time of $\cA$ for solving $\PSP(k, \cF^{\EQ}, N)$ is at most $O( N^{\delta k / 10})$ for every $N \in \N$.

Given a 3-CNF formula $\phi$ such that each variable appears in at most three clauses. Let $n, m$ denote the number of variables and the number of clauses of $\phi$, respectively. We first run the reduction from Proposition~\ref{prop:ETHPSP} on $\phi$ with this value of $k$. This produces an instance $(A_1, \dots, A_k)$ of \PSP$(k, \cF^{\EQ}, N)$ where $N = 2^{3\lceil m/k \rceil}$. Since each variable appears in at most three clauses, we have $m \leq 3n$, meaning that $N = O(2^{9n/k})$. By running $\cA$ on this instance, we can decide whether $\phi$ is satisfiable in time $O(N^{\delta k/10}) = O(2^{0.9 \delta n})$, contradicting \ETH.
\end{proof}

\section{Communication Protocols and Reduction to Gap Label Cover}
\label{sec:typeA:cc2maxcover}

In this section, we first introduce a communication model for multiparty communication known in literature as the Simultaneous Message Passing model. Then, we introduce a notion of ``efficient'' communication protocols, and connect the existence of such protocols to a reduction from \PSP to a gap version of \maxcover.

\subsection{Efficient Protocols in Simultaneous Message Passing Model}

The two-player Simultaneous Message Passing (\SMP) model was introduced by Yao \cite{Yao79} and has been extensively studied in literature \cite{KushilevitzNisan96-CCBook}. In the multiparty setting, the \SMP model is considered popularly with the number-on-forehead model, where each player can see the input of all the other players but not his own \cite{CFL83,BGKL03}. 
In this paper, we consider the multiparty \SMP model where the inputs are given as in the number-in-hand model (like in \cite{FOZ16,WW15}). 

{\bf Simultaneous Message Passing Model.} Let $f:\{0,1\}^{m\times k}\to\{0,1\}$.  In the $k$-player simultaneous message passing communication model, we have $k$ players each with an input $x_i\in \{0,1\}^m$ and a referee who is given an advice $\mu\in \{0,1\}^{*}$ (at the same time when the players are given the input). The communication task is for the referee to determine if $f(x_1,\ldots ,x_k)=1$. The players are allowed to only send messages to the referee. In the randomized setting, we allow the players \emph{and} the referee to jointly toss some random coins before sending messages, i.e., we allow public randomness.

Next, we introduce the notion of \emph{efficient} protocols, which are in a nutshell one-round randomized protocols where the players and the referee are in a computationally bounded setting.

{\bf Efficient Protocols.}
Let $\pi$ be a communication protocol for a problem in the \SMP model. We say that $\pi$ is  a $(w,r,\ell,s)$-efficient protocol if the following holds:
\begin{itemize}
\item The referee receives $w$ bits of advice.
\item The protocol is one-round with public randomness, i.e., the following actions happen sequentially:
\begin{enumerate}
\item The players receive their inputs and the referee receives his advice.
\item The players and the referee jointly toss $r$ random coins.
\item Each player on seeing the randomness (i.e. results of $r$ coin tosses) deterministically sends an $\ell$-bit message to the referee.
\item Based on the advice, the randomness, and the total $\ell\cdot k$ bits sent from the players, the referee outputs accept or reject.
\end{enumerate}
\item The protocol has completeness 1 and soundness $s$, i.e.,
\begin{itemize}
\item If $f(x_1,\ldots ,x_k)=1$, then there exists an advice on which the referee always accepts.
\item If $f(x_1,\ldots ,x_k)=0$, then, on any advice, the referee accepts with probability at most $s$.
\end{itemize}
\item The players and the referee are computationally bounded, i.e., all of them perform all their computations in $\poly(m)$-time.
\end{itemize}

The following proposition follows immediately from the definition of an efficient protocol and will be very useful in later sections for gap amplification. 

\begin{proposition}\label{prop:repeat}
Let $z\in\mathbb N$ and $\pi$ be a communication protocol for a problem in the \SMP model. Suppose $\pi$ is a $(w,r,\ell,s)$-efficient protocol. Then there exists a $(w,z\cdot r,z\cdot \ell,s^{z})$-efficient protocol for the same problem.
\end{proposition}
\begin{proof}
The proof follows by a simple repetition argument. More precisely, we repeat steps 2-4 in the protocol $z$ times, each time using fresh randomness, but note that the $z$ steps of drawing random coins can be clubbed into one step, and the decision by the referee can be reserved till the end of the entire protocol, wherein he accepts if and only if he would accept in each of the individual repetitions. 
\end{proof}

\subsection{Lower Bounds on Gap-\maxcover}\label{sec:gapmax}

The following theorem is the main conceptual contribution of the paper: we show below that the existence of efficient protocols can translate (exact) hardness of \PSP{}s to hardness of approximating \maxcover.

\begin{theorem} \label{thm:meta-maxcov}
Let $m: \N \times \N \to \N$ be any function. 
Let $\mathcal{F}:=\{f_{N, k}:\{0,1\}^{m(N, k)\times k}\to\{0,1\}\}_{N ,k \in \N}$ be a family of Boolean functions indexed by $N, k$. 
Suppose there exists a $(w,r,\ell,s)$-efficient protocol\footnote{$w, r, \ell$ and $s$ can depend on $N$ and $k$.} for  $f_{N, k}$ in the $k$-player \SMP model for every $N, k \in \N$.
Then, there is a reduction from any instance $(A_1,\dots, A_k)$ of $\PSP(k,\mathcal{F},N)$ to $2^w$ label cover instances $\{\Gamma^\mu\}_{\mu \in \{0, 1\}^w}$ such that
\begin{itemize}
\item The running time of the reduction is $2^{w + r + \ell k}\poly(m(N, k))$.
\item Each $\Gamma^{\mu} = (U^{\mu} = U^{\mu}_1 \cup \cdots \cup U^{\mu}_q, W^{\mu} = W^{\mu}_1 \cup \cdots \cup W^{\mu}_h; E^{\mu})$ has the following parameters:
\begin{itemize}
\item $\Gamma^\mu$ has at most $Nk$ right nodes, i.e., $|W^\mu| \leq Nk$,
\item $\Gamma^\mu$ has $k$ right super nodes, i.e., $h = k$,
\item $\Gamma^\mu$ has $2^r$ left super nodes, i.e., $q = 2^r$,
\item $\Gamma^\mu$'s left alphabet size is at most $2^{\ell k}$, i.e., $|U^\mu_1|, \dots, |U^\mu_q| \leq 2^{\ell k}$.
\end{itemize}
\item If $(A_1, \dots, A_k)$ is a YES instance of $\PSP(k, \mathcal{F}, N)$, then $\maxcover(\Gamma^\mu) = 1$ for some $\mu\in \{0, 1\}^w$.
\item If $(A_1, \dots, A_k)$ is a NO instance of $\PSP(k, \mathcal{F}, N)$, then $\maxcover(\Gamma^\mu) \leq s$ for every $\mu \in \{0, 1\}^w$. 
\end{itemize}
\end{theorem}

\begin{proof}
Given a $(w,r,\ell,s)$-efficient protocol $\pi$ of $f_{N, k}$ and an instance $(A_1, \dots, A_k)$ of $\PSP(k, \cF, N)$,
we will generate $2^w$ instances of \maxcover. Specifically, for each $\mu \in \{0, 1\}^w$, we construct an instance $\Gamma^\mu = (U^{\mu} = U^{\mu}_1 \cup \cdots \cup U^{\mu}_q, W^{\mu} = W^{\mu}_1 \cup \cdots \cup W^{\mu}_h; E^{\mu})$ of \maxcover as follows.

\begin{itemize}
\item Let $h = k$. For each $j \in [h]$, the right super-node $W^\mu_j$ contains one node for each $x_j \in A_j$.
\item Let $q = 2^r$. For each random string $\gamma \in \{0,1\}^r$, the left-super node $U^\mu_\gamma$ contains one node for each of the possible accepting messages from the $k$ players, i.e., each vertex in $U^\mu_\gamma$ corresponds to $(m_1, \dots, m_k) \in (\{0, 1\}^\ell)^k$ where in the protocol $\pi$ the referee, on an advice $\mu$ and a random string $\gamma$, accepts if the messages he received from the $k$ players are $m_1, \dots, m_k$.
\item We add an edge between $x_j \in W^\mu_j$ and $(m_1, \dots, m_k) \in U^\mu_\gamma$ if $m_j$ is equal to the message that $j$ sends on an input $x_j$ and a random string $\gamma$ in the protocol $\pi$.
\end{itemize}
Observe that there is a bijection between labelings of $\Gamma^\mu$
and elements of $A_1 \times \cdots \times A_k$.

Now consider a labeling $S\subseteq W^\mu$ of $\Gamma^\mu$ and the corresponding $(x_1, \dots, x_k) \in A_1 \times \cdots \times A_k$.
For each random string $\gamma \in\{0,1\}^{r}$,
observe that the referee accepts on an input $(x_1, \dots, x_k)$, an advice $\mu$, and a random string $\gamma$ if and only if there is a vertex $u\in U_\gamma$ (corresponding to the messages sent by the players) that has an edge to every vertex in $S$.
%
%Conversely, if there is a vertex $u\in U_\gamma$ that is incident to every vertex in $S$,
%then we know that $u$ is the concatenation of the $\ell$ messages corresponding to $S$
%which cause the referee to accept.
%
Therefore, the acceptance probability of the protocol on advice $\mu$ is 
the same as the fraction of left super-nodes covered by $S$.
The completeness and soundness then easily follows:

{\bf Completeness.}
If there exists $(x_1, \dots, x_k) \in A_1 \times \cdots \times A_k$ such that $f_{N,k}(x_1,\ldots,x_k)=1$, then there is an advice $\mu \in \{0, 1\}^w$
on which the referee always accepts for this input $(x_1, \dots, x_k)$,
meaning that the corresponding labeling covers every left super-node of $\Gamma^\mu$, i.e., $\maxcover(\Gamma^{\mu}) = 1$.

{\bf Soundness.}
If $f_{N,k}(x_1,\ldots,x_k) = 0$ for every $(x_1, \dots, x_k) \in A_1 \times \cdots \times A_k$, then, for any advice $\mu \in \{0, 1\}^w$,
the referee accepts with probability at most $s$ on every input $(x_1, \dots, x_k) \in A_1 \times \cdots \times A_k$. This means that, for any $\mu \in \{0, 1\}^w$, no labeling covers more than $s$ fraction of left the super-nodes. In other words, $\maxcover(\Gamma^\mu) \leq s$ for all $\mu \in \{0, 1\}^w$.  
\end{proof}

For the rest of this subsection, we will use the following shorthand. Let $\Gamma = (U = U_1 \cup \cdots \cup U_q, W = W_1 \cup \cdots \cup W_h; E)$ be a label cover instance, and we use the shorthand $\Gamma(N,k,r,\ell)$ to say that the label cover instance has the following parameters:
\begin{itemize}
\item $\Gamma$ has at most $Nk$ right nodes, i.e., $|W| \leq Nk$,
\item $\Gamma$ has $k$ right super nodes, i.e., $h = k$,
\item $\Gamma$ has $2^r$ left super nodes, i.e., $q = 2^r$,
\item $\Gamma$ has left alphabet size of at most $2^{\ell k}$, i.e., $|U_1|, \dots, |U_q| \leq 2^{\ell k}$.
\end{itemize}

The rest of this section is devoted to combining Theorem~\ref{thm:meta-maxcov} with the results in Section~\ref{sec:typeA}
to obtain conditional hardness for the gap-\maxcover problem, assuming that we have efficient protocols with certain parameters. These protocols will be devised in the three subsequent sections.

\paragraph{Understanding the Parameters.}  Before we state the exact dependency of parameters, let us first discuss some intuition behind it. First of all, if we start with an instance of $\PSP(k, \cF, N)$, Theorem~\ref{thm:meta-maxcov} will produce $2^{w}$ instances of $\Gamma(N, k, r, \ell)$. Roughly speaking, since we want the lower bounds from \PSP to translate to \maxcover, we would like the number of instances to be $N^{o(1)}$, meaning that we want $w = o(\log N)$. Recall that in all function families we consider $m = \Theta_k(\log N)$. Hence, this requirement is the same as $w = o_k(m)$. Moreover, we would like the instance size of $\Gamma(N, k, r, \ell)$ to also be $O_k(N)$, meaning that the number of left vertices, $2^{r + \ell k}$ has to be $O_k(N)$. Thus, it suffices to have a protocol where $r + \ell k = o_k(m)$.

If we additionally want the hardness to translate also to $k$-\domset, the parameter dependencies become more subtle. Specifically, applying Theorem~\ref{lem:red-maxcov-domset} to the \maxcover instances results in a blow-up of $\sum_{j \in [q]} k^{|U_j|} = 2^r \cdot k^{2^{\ell k}} = 2^{r + (\log k) \cdot 2^{\ell k}}$. We also want this to be at most $N^{o(1)}$, meaning that we need $r + (\log k) \cdot 2^{\ell k} = o(\log N) = o_k(m)$. In other words, it suffices for us to require that $\ell k < \nicefrac{(\log m)}{\beta}$ for some constant $\beta > 1$. The exact parameter dependencies are formalized below.

\subsubsection*{\underline{\SETH}}

\begin{corollary}\label{cor:labSETH}
For any $c\in\mathbb N$, let $ \cF^{\SD}_c$ be the family of Boolean functions as defined in Definition~\ref{def:SETHFam}. 
For every $\delta>0$, suppose there exists a $(w,r,\ell,s)$-efficient protocol for  $\SD_{m,k}$ in the $k$-player \SMP model for every $k \in \N$ and every $m\in\mathbb N$, such that $w\le \delta m$ and $r+\ell k=o_k(m)$.
Then, assuming \SETH, for every $\varepsilon>0$ and integer $k>1$, no $O(N^{ k (1- \varepsilon)})$-time algorithm  can distinguish between $\maxcover(\Gamma) = 1$
and $\maxcover(\Gamma) \leq s$ for any label cover instance $\Gamma(N,k,r,\ell)$ for all $N \in \N$.
Moreover, if $\ell<\nicefrac{(\log m)}{\beta\cdot k}$ for some constant $\beta>1$, then assuming \SETH, for every $\varepsilon>0$ and integer $k>1$,  no $O(N^{ k (1- \varepsilon)})$-time algorithm  can distinguish between $\domset(G)=k$
and $\domset(G) \geq \left(\frac{1}{s}\right)^{\frac{1}{k}}\cdot k$ for any graph $G$ with at most $O_k(N)$ vertices, for all $N\in\mathbb N$.
\end{corollary}
\begin{proof}
The proof of the first part of the theorem statement is by contradiction. Suppose there is an $O(N^{ k (1- \varepsilon)})$-time algorithm $\mathcal A$ for some fixed constant $\varepsilon>0$ and integer $k>1$ which can distinguish between $\maxcover(\Gamma) = 1$
and $\maxcover(\Gamma) \leq s$ for any label cover instance $\Gamma(N,k,r,\ell)$ for all $N \in \N$. From Proposition~\ref{prop:SETHPSP}, we have that there exists $c_\varepsilon\in\mathbb N$ such that no $O(N^{ k (1- \nicefrac{\varepsilon}{2})})$-time algorithm can solve $\PSP(k, \cF^{\SD}_{c_\varepsilon}, N)$ for all $N \in \N$. Fix $\delta=\nicefrac{\varepsilon}{3c_\varepsilon}$. Next, by considering Theorem~\ref{thm:meta-maxcov} for the case of $(w,r,\ell,s)$-efficient protocols, we have that there are $2^w$ label cover instances $\{\Gamma^\mu\}_{\mu\in\{0,1\}^w}$ which can be constructed in $2^{\delta m(1+o_{k}(1))}$ time. Note that  $2^{\delta m(1+o_{k}(1))}=N^{\nicefrac{k\varepsilon}{3}(1+o_{k}(1))}$ by our choice of $\delta$. Thus, we can run $\mathcal{A}$ on each $\Gamma^\mu$ and solve  $\PSP(k, \cF^{\SD}_{c}, N)$ for all $N,k \in \N$ in time less than $N^{k(1-\nicefrac{\varepsilon}{2})}$. This contradicts Proposition~\ref{prop:SETHPSP}.

To prove the second part of the theorem statement, we apply the reduction described in Theorem~\ref{lem:red-maxcov-domset} and note that $2^r=N^{o(1)}$ and $k^{2^{\ell k}}=2^{^{(\log_2 k) \cdot (m(N,k))^{1/\beta}}}=N^{o(1)}$.
\end{proof}

The proof of Theorem~\ref{thm:domSETH} follows by plugging in the parameters of the protocol described in Corollary~\ref{cor:CC} to the above corollary.

\subsubsection*{\underline{\ETH}}

\begin{corollary}\label{cor:labETH}
Let $ \cF^{\EQ}$ be the family of Boolean functions as defined in Definition~\ref{def:fam-eq}. 
Suppose there exists a $(w,r,\ell,s)$-efficient protocol for  $\EQ_{m,k,t}$ in the $k$-player \SMP model for every $k,t,m \in \N$ such that $w+r+\ell k=o_k(m)$.
Then, assuming \ETH, for any computable  function $T: \N \to \N$, there is no $T(k) \cdot N^{o(k)}$ time algorithm that can distinguish between $\maxcover(\Gamma) = 1$
and $\maxcover(\Gamma) \leq s$ for any label cover instance $\Gamma(N,k,r,\ell)$ for all $N,k \in \N$.
Moreover, if $\ell<\nicefrac{(\log m)}{\beta\cdot k}$ for some constant $\beta>1$, then assuming \ETH, for any computable  function $T: \N \to \N$, there is no $T(k) \cdot N^{o(k)}$ time algorithm that can distinguish between $\domset(G)=k$
and $\domset(G) \geq \left(\frac{1}{s}\right)^{\frac{1}{k}}\cdot k$ for any graph $G$ with at most $O_k(N)$ vertices, for all $N,k\in\mathbb N$.
\end{corollary}
\begin{proof}
The proof of the first part of the theorem statement is by contradiction.  Suppose there is an algorithm $\mathcal A$ running in time  $\widetilde T(k)\cdot N^{o(k)}$ for some computable  function $\widetilde T:\mathbb N\to\mathbb N$ that can distinguish between $\maxcover(\Gamma) = 1$
and $\maxcover(\Gamma) \leq s$ for any label cover instance $\Gamma(N,k,r,\ell)$ for all $N,k \in \N$. From Lemma~\ref{lem:eth-psp}, we have that for any computable  function $T: \N \to \N$, there is no $T(k) \cdot N^{o(k)}$ time algorithm that can solve $\PSP(k, \cF^{\EQ}, N)$ for every $N, k \in \N$.  Next, by considering Theorem~\ref{thm:meta-maxcov} for the case of $(w,r,\ell,s)$-efficient protocols, we have that there are $2^w$ label cover instances $\{\Gamma^\mu\}_{\mu\in\{0,1\}^w}$ which can be constructed in $2^{o_k(m)}$ time. Note that  $2^{o_k(m)}=O_k(N^{o(1)})$ by the choice of $m(N,k)$ in Definition~\ref{def:fam-eq}. Thus, we can run $\mathcal{A}$ on each $\Gamma^\mu$ and solve  $\PSP(k, \cF^{\EQ}, N)$ for all $N,k \in \N$ in time   $\widetilde T(k)\cdot N^{o(k)}$. This contradicts Lemma~\ref{lem:eth-psp}.

To prove the second part of the theorem statement, we apply the reduction described in Theorem~\ref{lem:red-maxcov-domset} and note that $2^r=N^{o(1)}$ and $k^{2^{\ell k}}=2^{^{(m(N,k))^{1/\beta}\cdot \log_2 k}}=N^{o(1)}$.
\end{proof}

The proof of Theorem~\ref{thm:domETH} follows by plugging in the parameters of the protocol described in Corollary~\ref{cor:rep-multeq-protocol1} to the above corollary.

\subsubsection*{\underline{\Wone}}

\begin{corollary}\label{cor:labW1}
Let $ \cF^{\EQ}$ be the family of Boolean functions as defined in Definition~\ref{def:fam-eq}. 
Suppose there exists a $(w,r,\ell,s)$-efficient protocol for  $\EQ_{m,k,t}$ in the $k$-player \SMP model for every $k,t,m \in \N$ such that $w+r+\ell k<\nicefrac{m}{tk}$.
Then, assuming \Wone, for any computable function $T:\mathbb N\to\mathbb N$, there is no  $T(k)\cdot \poly(N)$-time algorithm that can distinguish between $\maxcover(\Gamma) = 1$
and $\maxcover(\Gamma) \leq s$ for any label cover instance $\Gamma(N,k,r,\ell)$ for all $N,k \in \N$.
Moreover, if $r<\nicefrac{m}{2tk}$ and $\ell<\nicefrac{(\log m)}{\beta\cdot k}$ for some constant $\beta>1$, then assuming \Wone, for any computable function $T:\mathbb N\to\mathbb N$, there is no  $T(k)\cdot \poly(N)$-time algorithm that can  distinguish between $\domset(G)=k$
and $\domset(G) \geq \left(\frac{1}{s}\right)^{\frac{1}{k}}\cdot k$ for any graph $G$ with at most $O_k(N)$ vertices, for all $N,k\in\mathbb N$.
\end{corollary}
\begin{proof}
The proof of the first part of the theorem statement is by contradiction. Suppose there is an algorithm $\mathcal A$ running in time  $\widetilde T(k)\cdot \poly(N)$ for some computable  function $\widetilde T:\mathbb N\to\mathbb N$ that can distinguish between $\maxcover(\Gamma) = 1$
and $\maxcover(\Gamma) \leq s$ for any label cover instance $\Gamma(N,k,r,\ell)$ for all $N,k \in \N$. From Lemma~\ref{lem:w1-psp}, we have that for any computable function $T: \N \to \N$, there is no $T(k) \cdot \poly(N)$ time algorithm that can solve $\PSP(k, \cF^{\EQ}, N)$ for every $N, k \in \N$.  Next, by considering Theorem~\ref{thm:meta-maxcov} for the case of $(w,r,\ell,s)$-efficient protocols, we have that there are $2^w$ label cover instances $\{\Gamma^\mu\}_{\mu\in\{0,1\}^w}$ which can be constructed in $2^{\nicefrac{m(N,k)}{k\cdot t(k)}}\cdot \poly(m(N,k))$ time. Note that  $2^{\nicefrac{m(N,k)}{k\cdot t(k)}}=O(N)$ and $\poly(m(N,k))=N^{o(1)}$ by the choice of $m(N,k)$ in Definition~\ref{def:fam-eq}. Thus, we can run $\mathcal{A}$ on each $\Gamma^\mu$ and solve  $\PSP(k, \cF^{\EQ}, N)$ for all $N,k \in \N$ in time less than $\widetilde T(k)\cdot \poly(N)$. This contradicts Lemma~\ref{lem:w1-psp}.

To prove the second part of the theorem statement, we apply the reduction described in Theorem~\ref{lem:red-maxcov-domset} and note that $2^r=O(\sqrt N)$ and $k^{2^{\ell k}}=2^{^{(m(N,k))^{1/\beta}\cdot \log_2 k}}=N^{o(1)}$.
\end{proof}

The proof of Theorem~\ref{thm:domW1} follows by plugging in the parameters of the protocol described in Corollary~\ref{cor:rep-multeq-protocol1} to the above corollary.

\subsubsection*{\underline{\ksum Hypothesis}}

%Finally, the lower bound on gap-\maxcover by assuming the \ksum hypothesis is stated in the following corollary which easily follows by putting together Theorem~\ref{thm:meta-maxcov} and Proposition~\ref{prop:ksumPSP}. We omit the proof as it is very similar to the previous proofs.	

\begin{corollary}\label{cor:labksum}
Let $ \cF^{\SZ}$ be the family of Boolean functions as defined in Definition~\ref{def:ksumFam}. 
Suppose there exists a $(w,r,\ell,s)$-efficient protocol for  $\SZ_{m,k}$ in the $k$-player \SMP model for every $m, k \in \N$, such that $w+r+\ell k=o_k(m)$.
Then assuming the \ksum Hypothesis, for every integer $k \geq 3$ and every $\varepsilon > 0$, no $O(N^{\lceil k/2 \rceil - \varepsilon})$-time algorithm can distinguish between $\maxcover(\Gamma) = 1$
and $\maxcover(\Gamma) \leq s$ for any label cover instance $\Gamma(N,k,r,\ell)$ for all $N \in \N$.
Moreover, if $\ell<\nicefrac{(\log m)}{\beta\cdot k}$ for some constant $\beta>1$, then assuming the \ksum Hypothesis, for every $\varepsilon>0$  no $O(N^{  {\lceil k/2 \rceil - \varepsilon}})$-time algorithm  can distinguish between $\domset(G)=k$
and $\domset(G) \geq \left(\frac{1}{s}\right)^{\frac{1}{k}}\cdot k$ for any graph $G$ with at most $O_k(N)$ vertices for all $N\in\mathbb N$.
\end{corollary}
\begin{proof}
The proof of the first part of the theorem statement is by contradiction.  Suppose there is an algorithm $\mathcal A$ running in time  $O(N^{\lceil k/2 \rceil - \varepsilon}) $ for some fixed constant $\varepsilon>0$ and some integer $k\ge 3$ that can distinguish between $\maxcover(\Gamma) = 1$
and $\maxcover(\Gamma) \leq s$ for any label cover instance $\Gamma(N,k,r,\ell)$ for all $N \in \N$. From Proposition~\ref{prop:ksumPSP}, we have that no $O(N^{\lceil k/2 \rceil - \varepsilon/2})$-time algorithm can solve $\PSP(k, \cF^{\SZ}, N)$ for all $N \in \N$.  Next, by considering Theorem~\ref{thm:meta-maxcov} for the case of $(w,r,\ell,s)$-efficient protocols, we have that there are $2^w$ label cover instances $\{\Gamma^\mu\}_{\mu\in\{0,1\}^w}$ which can be constructed in $2^{o_k(m)}$ time. Note that  $2^{o_k(m)}=O_k(N^{o(1)})$ by the choice of $m(N,k)$ in Definition~\ref{def:ksumFam}. Thus, we can run $\mathcal{A}$ on each $\Gamma^\mu$ and solve  $\PSP(k, \cF^{\SZ}, N)$ for all $N \in \N$ in time   $O(N^{\lceil k/2 \rceil - \varepsilon})$. This contradicts Proposition~\ref{prop:ksumPSP}.

To prove the second part of the theorem statement, we apply the reduction described in Theorem~\ref{lem:red-maxcov-domset} and note that $2^r=N^{o(1)}$ and $k^{2^{\ell k}}=2^{^{(m(N,k))^{1/\beta}\cdot \log_2 k}}=N^{o(1)}$.
\end{proof}

The proof of Theorem~\ref{thm:domksum} follows by plugging in the parameters of the protocol described in Corollary~\ref{cor:rep-ksum-protocol} to the above corollary.

\section{An Efficient Protocol for Set Disjointness}
\label{sec:SETH}

Set Disjointness has been extensively studied primarily in the two-player setting (i.e., $k=2$). In that setting, we know that the randomized communication complexity is $\Omega(m)$  \cite{KS92,R92,BJKS04}, where $m$ is the input size of each player. Surprisingly, \cite{AW09} showed that the \MA-complexity of two-player set disjointness is $\tilde{O}(\sqrt m)$. Their protocol was indeed an $(\tilde{O}(\sqrt{m}),O(\log m),\tilde{O}(\sqrt{m}),\nicefrac{1}{2})$-efficient protocol for the case  when $k=2$. Recently, \cite{ARW17} improved (in terms of the message size) the protocol to be an $(\nicefrac{m}{\log m},O(\log m),O((\log{m})^3),\nicefrac{1}{2})$-efficient protocol for the case when $k=2$. Both these results can be extended naturally for all $k>1$, to give an $(\nicefrac{mk}{\log m},O_k(\log m),O_k((\log{m})^3),\nicefrac{1}{2})$-efficient protocol. However, this does not suffice for proving Theorem~\ref{thm:domSETH} since we need a $(w,r,\ell,s)$-efficient protocol with $w=o(m)$ \emph{and} $\ell=o(\log m)$. Fortunately, Rubinstein \cite{R17} recently showed that the exact framework of the \MA-protocols as in $\cite{AW09,ARW17}$ but with the use of algebraic geometric codes instead of Reed Muller or Reed Solomon codes gives the desired parameters in the two-player case. Below we naturally extend Rubinstein's protocol to the $k$-player setting. This extension was suggested to us  by Rubinstein~\cite{Rub17com}.

\begin{theorem}\label{thm:CC} 
There is a polynomial function $\hat{\ell}:\mathbb N\times \mathbb N\to [1,\infty)$ such that 
for every $k\in\mathbb N$ and every $\alpha\in\mathbb N$, there is a protocol for $k$-player $\SD_{m,k}$ in the \SMP model which is an $( \nicefrac{m}{\alpha},\log_2 m,\hat{\ell}(k,\alpha),\nicefrac{1}{2})$-efficient protocol, where each player is given $m$ bits as input, and the referee is given at most $\nicefrac{m}{\alpha}$ bits of advice. 
\end{theorem}

\begin{proof} Fix $k,m,\alpha\in\mathbb N$.
Let $q$ be the smallest prime greater than $\hat{q}(k)$ such that $q>\left(2\alpha\hat{r}(k)\right)^2$, where the functions $\hat{r}$ and $\hat{q}$ are as defined in Theorem~\ref{thm:ag-code}. Let $\mathcal{G}=\mathbb F_{q^2}$. Let $T={2\alpha\hat{r}(k)\log_2 q}$.

We associate $[m]$ with $[T] \times [m/T]$ and write the input $\mathbf x_j\in\{0,1\}^m$ as vectors $\mathbf{x}_j^1, \dots, \mathbf{x}_j^T$ where $\mathbf{x}_j^t\in\{0,1\}^{m/T}$. 
For every $j\in[k]$, Player $j$ computes $A_{m/T}(\mathbf{x}_j^t)$ for every $t\in [T]$. We denote the block length of $A_{m/T}$ by $d$. From the systematicity guaranteed by Theorem~\ref{thm:ag-code} we have that $A_{m/T}(\mathbf{x}_j^t)\mid_{[m/T]}=\mathbf{x}_j^t$.
Also, notice that for all $t\in[T]$, we have $\underset{j\in[k]}{\wedge}\mathbf{x}_j^t=0^{m/T}$ if and only if $\prod_{j \in [k]} A_{m/T}(\mathbf{x}_j^t)=0^{m/T}$. 

With the above observation in mind, we define the marginal sum $\Gamma\in \mathcal{G}^{d}$ as follows:
\begin{align*}
\forall i\in [d],\  \Gamma_i = \sum_{t \in [T]} \prod_{j \in [k]} A_{m/T}(\mathbf{x}_j^t)_i.
\end{align*}

Again, notice that for all $t\in[T]$, we have $\underset{j\in[k]}{\wedge}\mathbf{x}_j^t=0^{m/T}$ if and only if $\Gamma_i=0$ for all $i\in [m/T]$. This follows from the following:
\begin{itemize}
\item For all $i\in [m/T]$, we have $A_{m/T}(\mathbf{x}_j^t)_i\in\{0,1\}$ and thus $\underset{j \in [k]}{\prod} A_{m/T}(\mathbf{x}_j^t)_i\in\{0,1\}$.
\item The characteristic of $\mathcal G$ being greater than $\left({2\alpha\hat{r}(k)}\right)\cdot \sqrt{q}\ge \left({2\alpha}\hat{r}(k)\right)\cdot \log_2{q}=T$.
\end{itemize}

More importantly, we remark that $\Gamma$ is a codeword\footnote{We would like to remark that we use the multiplicity of Algebraic Geometric codes to find a non-trivial advice. On a related note, Meir \cite{M13} had previously shown that error correcting codes with the multiplicity property suffice to show the \textsf{IP} theorem (i.e., the \textsf{IP$=$PSPACE} result).} in $B_{m/T}$. This follows because $B_{m/T}$ is a degree $k$ closure code of $A_{m/T}$. To see this, in Definition~\ref{def:poly}, set $t=k$, $r=k\cdot T$, and $P[x_{1,1},\ldots ,x_{k,T}]=\sum_{i\in[T]}\prod_{j\in[k]}x_{i,j}$.
\subsubsection*{The protocol}

\begin{enumerate}
\item Merlin sends the referee $\Phi$ which is allegedly equal to the marginal sums codeword $\Gamma$ defined above.
\item All players jointly  draw $r \in [d]$ uniformly at random.
\item For every $j\in\{1,\ldots ,k\}$, Player $j$ sends to the referee $A_{m/T}(\mathbf{x}_j^t)_r$, $\forall t \in [T]$.
\item The referee accepts if and only if both of the following hold:
\begin{gather}
\forall  i \in [m/T], \;\;\ \Phi_i=0 \label{eq:Phi}\\
\ \Phi_r = \sum_{t \in [T]} \prod_{j \in [k]} A_{m/T}(\mathbf{x}_j^t)_r .\label{eq:test2}
\end{gather}
\end{enumerate}

\subsubsection*{Analysis}

{\bf Advice Length.} To send the advice, Merlin only needs to send a codeword in $B_{m/T}$ to the verifier. This means that the advice length (in bits) is no more than $\log_2 q^2$ times the block length, which is $d\le \left(\nicefrac{m}{T}\right)\hat{r}(k)=\nicefrac{m}{2\alpha\log_2q}$ where the equality comes from our choice of $T$ and the inequality from Theorem~\ref{thm:ag-code}.

{\bf Message Length.} Each player sends $T$ elements of $\mathcal{G}$. Hence, the message length is $T\log_2 q^2 \leq {\alpha\hat{r}(k)(2\log_2 q)^2} $. Recall that $q$ can be upper bounded by a  polynomial in $k$ and $\alpha$. Hence, the message length is upper bounded  entirely as a polynomial in $k$ and $\alpha$ as desired.

{\bf Randomness.} The number of coin tosses is $\log_2(d) \leq \log_2\left(\nicefrac{m\hat{r}(k)}{T}\right)=\log_2\left(\nicefrac{m}{2\alpha\log_2q}\right)< \log_2m$.

{\bf Completeness.}
If the $k$ sets are disjoint, Merlin can send the true $\Gamma$, and the verifier always accepts.

{\bf Soundness.}
If the $k$ sets are not disjoint and $\Phi$ is actually $\Gamma$, then \eqref{eq:Phi} is false and the verifier always rejects.
On the other hand, if $\Phi \ne \Gamma$, then, since both are codewords of $B_{m/T}$, from Theorem~\ref{thm:ag-code} their relative distance must be at least $\nicefrac{1}{2}$. As a result, with probability at least $\nicefrac{1}{2}$, $\Phi_r \ne \Gamma_r$. Since the right hand side of \eqref{eq:test2} is simply $\Gamma_r$, the verifier will reject for such $r$. Hence, the rejection probability is at least $\nicefrac{1}{2}$.
\end{proof}

The following corollary follows immediately by applying Proposition~\ref{prop:repeat} with $z=\nicefrac{(\log_2 m)}{2k\cdot \ell(k,\alpha)}$ to the above theorem.

\begin{corollary}\label{cor:CC}
There is a polynomial function $\hat{s}:\mathbb N\times \mathbb N\to [1,\infty)$ such that 
for every $k\in\mathbb N$ and every $\alpha\in\mathbb N$ there is a protocol for $k$-player $\SD_{m,k}$ in the \SMP model which is an $\left( \nicefrac{m}{\alpha},O\left((\log_2 m)^2\right),\nicefrac{(\log_2 m)}{2k},\left(\nicefrac{1}{m}\right)^{\nicefrac{1}{\hat{s}(k,\alpha)}}\right)$-efficient protocol, where each player is given $m$ bits as input, and the referee is given at most $\nicefrac{m}{\alpha}$ bits of advice. 
\end{corollary}

\section{An Efficient Protocol for \multeq} 
\label{sec:multeq}

\eq has been extensively studied, primarily in the two-player setting (i.e., $k=2$). In that setting, when public randomness is allowed, we know that the randomized communication complexity is $O(1)$  \cite{Yao79,KushilevitzNisan96-CCBook}, and the protocols can be naturally extended to the $k$-player \SMP model that yields a randomized communication complexity of $O(k)$. There are many protocols which achieve this complexity bound but for the purposes of proving Theorems~\ref{thm:domW1}~and~\ref{thm:domETH}, we will use the protocol where the players encode their input using a fixed good binary code and then send a jointly agreed random location of the encoded input to the referee who checks if all the messages he received are equal. Below we extend that protocol for \multeq.

\begin{theorem}\label{thm:multeqProtocol}
For some absolute constant $\delta > 0$, for every $t, k \in \N$ and every $m \in \N$ such that $m$ is divisible by $t$, there is a $(0, \log m + O(1), 2t, 1 - \delta)$-efficient protocol for $\EQ_{m, k, t}$ in the $k$-player \SMP model.
\end{theorem}

\begin{proof}
Let $\mathcal C=\{C_m:\{0,1\}^m\to \{0,1\}^{d(m)}\}_{m\in\N}$ be the family of good codes with rate at least $\rho$ and relative distance at least $\delta$ as guaranteed by Fact~\ref{lem:good-code}. Fix $m,k,t\in\mathbb N$ as in the theorem statement. 
Let $m' := \nicefrac{m}{t}$.

\subsubsection*{The protocol}

\begin{enumerate}
\item All players jointly draw $i^\star \in [d(m')]$ uniformly at random.
\item If player $j$'s input is $x_j = (x_{j,1}, y_{j,1}, \ldots, x_{j,t}, y_{j,t})$, then he sends $(C(x_{j,1})_{i^*}, y_{j,1}, \ldots, C(x_{j,t})_{i^*}, y_{j,t})$ to the referee.
\item The referee accepts if and only if the following holds:
$$\tEQ_{2t, k,t}((C(x_{1,1})_{i^*}, y_{1,1}, \ldots, C(x_{1,t})_{i^*}, y_{1,t}), \ldots, (C(x_{k,1})_{i^*}, y_{k,1}, \ldots, C(x_{k,t})_{i^*}, y_{k,t})) = 1.$$
\end{enumerate}

\subsubsection*{Analysis}

{\bf Parameters of the Protocol.}
In the first step of the protocol all the players jointly draw $i^\star \in [d(m')]$ uniformly, which requires $\lceil\log d(m')\rceil \le  \log m' + \log (1/\rho) \le \log m + O(1)$ random bits. Then, for every $j\in[k]$, player $j$ sends the referee $2t$ bits. Finally, since the code $C_{m'}$ is efficient, it is easy to see that the players and referee run in $\poly(m)$-time.

{\bf Completeness.}
If $\EQ_{m,k,t}(x_1, \ldots, x_k) = 1$, then, for every $q \in [t]$ and $i, j \in [k]$, we have $(y_{i,q} = \bot) \vee (y_{j,q} = \bot) \vee (x_{i,q} = x_{j,q})=1$. This implies that, for every $q \in [t], i, j \in [k]$ and $i^* \in [d(m)]$, we have $(y_{i,q} = \bot) \vee (y_{j,q} = \bot) \vee (C(x_{i,q})_{i^*} = C(x_{j,q})_{i^*})$. In other words, $\tEQ_{2t, k}((C(x_{1,1})_{i^*}, y_{1,1}, \ldots, C(x_{1,t})_{i^*}, y_{1,t}), \ldots, (C(x_{k,1})_{i^*}, y_{k,1}, \ldots, C(x_{k,t})_{i^*}, y_{k,t})) = 1$ for every $i^* \in [d(m)]$, meaning that the referee always accepts.

{\bf Soundness.}
Suppose that $\tEQ_{m,k,t}(x_1, \ldots, x_k) = 0$. Then, there exists some $q \in [t]$ and $i, j \in [k]$ such that $y_{i,q} = \top, y_{j,q} = \top$ and $x_{i,q} \ne x_{j,q}$. Since the code $C$ has relative distance at least $\delta$, $C(x_{i,q})$ and $C(x_{j,q})$ must differ on at least $\delta$ fraction of the coordinates. When the randomly selected $i^*$ is such a coordinate, we have that $\tEQ_{2t, k}((C(x_{1,1})_{i^*}, y_{1,1}, \ldots, C(x_{1,t})_{i^*}, y_{1,t}), \ldots, (C(x_{k,1})_{i^*}, y_{k,1}, \ldots, C(x_{k,t})_{i^*}, y_{k,t})) = 0$, i.e., the referee rejects. Hence, the referee rejects with probability at least $\delta$.
\end{proof}

The following corollary follows immediately by applying Proposition~\ref{prop:repeat}  to the above theorem with $z=\nicefrac{\log_2 m}{4kt}$.

\begin{corollary} \label{cor:rep-multeq-protocol1}
For every $t, k \in \N$ and every $m \in \N$ such that $m$ is divisible by $t$, there is a $(0, O((\log m)^2), \nicefrac{(\log_2 m)}{2k}, \left(\nicefrac{1}{m}\right)^{\nicefrac{1}{O(kt)}})$-efficient protocol for $\EQ_{m, k, t}$ in the $k$-player \SMP model.\end{corollary}

\section{An Efficient Protocol for \sumz}
\label{sec:ksum}

The \sumz problem has been studied in the \SMP model and efficient protocols with the following parameters have been obtained.

\begin{theorem}[\cite{Nisan94}]\label{thm:Nisan}
For every $k\in\mathbb N$ and $m\in\mathbb N$ there is a $(0, O(\log (m+\log k)) ,O(\log (m+\log k)), \nicefrac{1}{2})$-efficient protocol for $\SZ_{m, k}$ in the $k$-player \SMP model.
\end{theorem}

The above protocol is based on a simple (yet powerful) idea of picking a small random prime $p$ and checking if the numbers sum to zero modulo $p$. Viola put forth a protocol with better parameters than the above protocol (i.e., smaller message length) using specialized hash functions and obtained the following:

\begin{theorem}[\cite{Viola15}]\label{thm:Viola}
For every $k\in\mathbb N$ and $m\in\mathbb N$ there is a $(0, O(m) ,O(\log k), \nicefrac{1}{2})$-efficient protocol for $\SZ_{m, k}$ in the $k$-player \SMP model.
\end{theorem}

%\knote{ How to cite  \cite{DHKP97} \cite{CFL83} }

In order to prove Theorem~\ref{thm:domksum}, we need a protocol with $o(m)$ randomness \emph{and} $o(\log m)$ message length, and both the protocols described above do not meet these conditions. 
We show below that the above two results can be composed to get a protocol with $O(\log k)$ communication complexity and $O_k(\log m)$ randomness. In fact, we will use a slightly different protocol from~\cite{Viola15}: namely, the protocol for the \sumzp problem as stated below.
%In particular, we will use a slightly different result than Theorem~\ref{thm:Viola} from \cite{Viola15}. 

\begin{definition}[\sumzp Problem]
Let $k, m,p \in \N$. $\SEp_{m, k}: (\{0, 1\}^m)^k \rightarrow \{0, 1\}$ is defined by
\begin{align*}
\SEp_{m, k}(x_1, \dots, x_k) = \begin{cases}
1\text{ if } \sum_{i\in [k]}x_i=0\mod p,\\
0\text{ otherwise},
\end{cases}
\end{align*}
where we think of each $x_i$ as a number in $[-2^{m-1},2^{m-1}-1]$.
\end{definition}

\begin{theorem}[\cite{Viola15}]\label{thm:Viola1}
For every $p,k\in\mathbb N$ and $m\in\mathbb N$, there is a $(0, O(\log p) ,O(\log k), \nicefrac{1}{2})$-efficient protocol\footnote{As written in \cite{Viola15}, the protocol has the following steps. First, the players send some messages to the referee. Then the players and the referee jointly draw some random coins, and finally the players send some more messages to the referee. However, we note that the first and second steps of the protocol can be swapped in \cite{Viola15} to obtain an efficient protocol as the drawing of randomness do not depend on the messages sent by the players to the referee.} for $\SEp_{m, k}$ in the $k$-player \SMP model.
\end{theorem}

Below is our theorem which essentially combines Theorem~\ref{thm:Nisan} and Theorem~\ref{thm:Viola1}. 
\begin{theorem}\label{thm:RandKSum}
For every $k\in\mathbb N$ and $m\in\mathbb N$ there is a $(0, O(\log(m+\log k)) ,O(\log k), \nicefrac{3}{4})$-efficient protocol for $\SZ_{m, k}$ in the $k$-player \SMP model.
\end{theorem}
\begin{proof}
Let $\pi^*$ be the protocol of Viola from Theorem~\ref{thm:Viola1}.

\subsubsection*{The protocol}

Let $t=2(m-1+\log k)$. Let $p_1,\ldots ,p_t$ be the first $t$ primes.

\begin{enumerate}
\item All players jointly  draw $i^\star \in [t]$ uniformly at random.
\item The players and the referee run $\pi^\star$ where each player now has input $y_i=x_i\mod p_{i^\star}$, and the referee accepts if and only if $\sum_{i\in [k]}y_i=0\mod p_{i^\star}$.
\end{enumerate}

Notice that the above protocol is still an efficient protocol as the first step of the above protocol can be combined with the draw of random coins in the first step of $\pi^\star$ to form a single step in which the players and the referee jointly draw random coins from the public randomness.
\subsubsection*{Analysis}

{\bf Randomness.} In the first step of the protocol all the players jointly draw $i^\star \in [t]$ uniformly, which requires $\lceil\log_2 t\rceil$ random bits. Then, the players draw $O(\log p_{i^\star})$ additional random coins for $\pi^{\star}$. The bound on the randomness follows by noting that $p_{i^\star}\le p_t=O(t\log t)$ .

{\bf Message Length.} For every $j\in[k]$, Player $j$ sends the referee $O(\log k)$ bits as per $\pi^\star$. 

{\bf Completeness.}
If the $k$ numbers sum to zero, 
then for any $p\in \mathbb N$, they sum to zero mod $p$ and thus the referee always accepts.

{\bf Soundness.}
If the $k$ numbers do not sum to zero, then let $\textsf{s}(\pi^\star)$ be the soundness of $\pi^\star$ and let $S$ be the subset of the first $t$ primes defined as follows:
$$
S=\left\{p_i\Bigg| i\in [t],\sum_{j\in[k]}x_j=0\text{ mod }p_{i}\right\}
$$
It is clear that the referee rejects with probability at least $\left(1-\nicefrac{|S|}{t}\right)\cdot (1-\textsf{s}(\pi^\star))$.
Therefore, it suffices to show that $|S|< \nicefrac{t}{2}$ as $\textsf{s}(\pi^\star)\le \nicefrac{1}{2}$.
Let $x:=\sum_{j\in[k]}x_j$.
We have that $x\in [-k\cdot 2^{m-1},k\cdot 2^{m-1}]$ and that, for every $p\in S$, $p$ divides $x$.
Since $x\neq 0$, we know that $|x|\ge \underset{p\in S}{\prod}\ p\ge \underset{i=1}{\overset{|S|}{\prod}}\ p_i\ge 2^{|S|}$.
Since $|x|\le k\cdot 2^{m-1}$, we have that $|S|< m-1+\log k=t$, and the proof follows.
\end{proof}

The following corollary follows immediately by applying Proposition~\ref{prop:repeat} with $z=\nicefrac{\log_2 m}{2k\cdot c\log k}$ to the above theorem, where $c$ is some constant such that the message length of the protocol in Theorem~\ref{thm:RandKSum} is at most $c\log k$.

\begin{corollary} \label{cor:rep-ksum-protocol}
For every $k,m\in\mathbb N$ there is a $\left(0,O((\log (m+\log k))^2),\nicefrac{(\log_2 m)}{2k},\left(\nicefrac{1}{m}\right)^{\nicefrac{1}{O(k\log k)}}\right)$-efficient protocol for $\SZ_{m, k}$ in the $k$-player \SMP model.
\end{corollary}

\section{Conclusion and Open Questions} \label{sec:conclusion}

We showed the parameterized inapproximability results for $k$-\domset 
under \Wone, \ETH, \SETH and \ksum Hypothesis, 
which almost resolve the complexity status of approximating parameterized $k$-\domset.
Although we showed the $\W[1]$-hardness of the problem, the exact version of $k$-\domset is $\W[2]$-complete.
Thus, a remaining question is whether approximating $k$-\domset is $\W[2]$-hard:
\begin{question}
Can we base total inapproximability of $k$-\domset on $\W[2] \ne \FPT$?
\end{question}

We note that even $1.01$-approximation of $k$-\domset is not known to be $\W[2]$-hard.

Another direction is to look beyond $k$-\domset and try to prove inapproximability of other parameterized problems. Since $k$-\domset is $\W[2]$-complete, there are already known reductions from it to other $\W[2]$-hard problems. Some of these reductions such as those for $k$-Set Cover, $k$-Hitting Set, $k$-\domset on tournaments~\cite{DowneyF1995} and $k$-Contiguous Set~\cite{CharikarNW16} are gap-preserving reductions (in the sense of \cite[Definition 3.4]{ChalermsookCKLMNT17}) and total inapproximability under \Wone translate directly to those problems. In this sense, one may wish to go beyond $\W[2]$-hard problems and prove hardness of approximation for problems in $\W[1]$. One of the most well-studied $\W[1]$-complete problems is $k$-\clique which we saw earlier in our proof; in its optimization variant, one wishes to find a clique of maximum size in the graph. Note that in this case a $k$-approximation is trivial since we can just output one vertex. It was shown in \cite{ChalermsookCKLMNT17} that no $o(k)$-approximation is possible in FPT time assuming \gapETH. Since this rules out all non-trivial approximation algorithms, such a non-existential result is also referred to as the \emph{total FPT-inapproximability} of $k$-\clique. Hence, the question is whether one can bypass \gapETH for this result as well:

%Note that while we know that $\mathrm{W}[1]\subseteq\mathrm{W}[2]$,
%there is no evidence that $\mathrm{[W}[1]\neq\mathrm{W}[2]$,
%and there is no know consequence if the $\mathrm{W}$-hierarchy collapse.

%Since $k$-\domset is $\W[2]$-complete, ones may think of $k$-\domset as the representative of $\W[2]$-hard problems. In this sense, we would wish to establish inapproximability results for $k$-\clique, which can be thought of as the representative of $\W[1]$-hard problems.

%As we discussed, the $\mathrm{W}$-hierarchy,
%ones may think of $k$-\domset as the representative of $\mathrm{W}[2]$-hard problems.
%In this sense, we would wish to establish the same inapproximability results
%for $k$-\clique, which can be thought of as the representative of $\mathrm{W}[1]$-hard problems.
%
%It was shown in \cite{ChalermsookCKLMNT17} under Gap-ETH that $k$-\clique admits no
%$F(k)$-approximation in $n^{o(k)}$ time, for any function $F(k)=o(k)$
%(thus, ruling out FPT-approximation algorithms).
%
%A straightforward question is whether ones can bypass Gap-ETH.
%In particular, we wish to prove the same result under ETH or $\mathrm{W}[1] \neq \mathrm{FPT}$.

\begin{question}
Can we base total inapproximability of $k$-\clique on \Wone or \ETH?
\end{question}

Again, we note that the current state-of-the-art results do not even rule out $1.01$-FPT-approximation algorithms for $k$-\clique under \Wone or \ETH.

Another direction is to look beyond parameterized complexity questions. As mentioned earlier, Abboud et al.~\cite{ARW17} used the hardness of approximating of \PCP-Vectors as a starting point of their inapproximability results of problems in \P. Since \maxcover is equivalent to \PCP-Vectors when the number of right super-nodes is two, it may be possible that \maxcover for larger number of right super-nodes can also be used to prove hardness of problems in \P\ as well. At the moment, however, we do not have any natural candidate in this direction (see Appendix~\ref{sec:fine} for further discussions).

\medskip

\paragraph*{Acknowledgment.} We would like to express our sincere gratitude and appreciation to Aviad Rubinstein
for suggesting the use of algebraic geometric codes, which resolved our issues 
in constructing the protocol for Set Disjointness with our desired parameters.
We would like to thank Swastik Kopparty, Henning Stichtenoth, Gil Cohen, Alessandro Chiesa and Nicholas Spooner for very helpful discussions on  Algebraic Geometric Codes, Emanuele Viola for insightful comments on his protocol from~\cite{Viola15}, and Guy Rothblum for discussions on locally characterizable sets. 
We also thank Irit Dinur, Prahladh Harsha and Yijia Chen for general discussions and suggestions.

Karthik C.\ S.\ is supported by ERC-CoG grant no. 772839 and ISF-UGC grant no. 1399/14.
Bundit Laekhanukit is partially supported by ISF grant no. 621/12, I-CORE grant no. 4/11
and by the DIMACS/Simons Collaboration on Bridging Continuous and Discrete Optimization through NSF grant no. CCF-1740425.
Pasin Manurangsi is supported by NSF grants no. CCF 1540685 and CCF 1655215, and by ISF-UGC grant no. 1399/14.
Parts of the work were done while all the authors were at the Weizmann Institute of Science
and while the second author was visiting the Simons Institute for the Theory of Computing.

\bibliographystyle{alpha}
\bibliography{ref}

\newcommand{\etalchar}[1]{$^{#1}$}
\begin{thebibliography}{DFMR08}

\bibitem[ABC09]{ABC09}
Chrisil Arackaparambil, Joshua Brody, and Amit Chakrabarti.
\newblock Functional monitoring without monotonicity.
\newblock In {\em ICALP}, pages 95--106, 2009.

\bibitem[AL13]{AL13}
Amir Abboud and Kevin Lewi.
\newblock Exact weight subgraphs and the k-sum conjecture.
\newblock In {\em ICALP}, pages 1--12, 2013.

\bibitem[ALM{\etalchar{+}}98]{ALMSS98}
Sanjeev Arora, Carsten Lund, Rajeev Motwani, Madhu Sudan, and Mario Szegedy.
\newblock Proof verification and the hardness of approximation problems.
\newblock {\em J. {ACM}}, 45(3):501--555, 1998.

\bibitem[ALW14]{AbboudLW14}
Amir Abboud, Kevin Lewi, and Ryan Williams.
\newblock Losing weight by gaining edges.
\newblock In {\em ESA}, pages 1--12, 2014.

\bibitem[AMS06]{AlonMS06}
Noga Alon, Dana Moshkovitz, and Shmuel Safra.
\newblock Algorithmic construction of sets for \emph{k}-restrictions.
\newblock {\em {ACM} Trans. Algorithms}, 2(2):153--177, 2006.

\bibitem[AMS12]{AMS12}
Noga Alon, Ankur Moitra, and Benny Sudakov.
\newblock Nearly complete graphs decomposable into large induced matchings and
  their applications.
\newblock In {\em STOC}, pages 1079--1090, 2012.

\bibitem[ARW17]{ARW17}
Amir Abboud, Aviad Rubinstein, and Ryan Williams.
\newblock Distributed {PCP} theorems for hardness of approximation in {P}.
\newblock In {\em FOCS}, pages 25--36, 2017.

\bibitem[AS98]{AS98}
Sanjeev Arora and Shmuel Safra.
\newblock Probabilistic checking of proofs: {A} new characterization of {NP}.
\newblock {\em J. {ACM}}, 45(1):70--122, 1998.

\bibitem[AW09]{AW09}
Scott Aaronson and Avi Wigderson.
\newblock Algebrization: {A} new barrier in complexity theory.
\newblock {\em {TOCT}}, 1(1):2:1--2:54, 2009.

\bibitem[BCG{\etalchar{+}}16]{BCGRS17}
Eli Ben{-}Sasson, Alessandro Chiesa, Ariel Gabizon, Michael Riabzev, and
  Nicholas Spooner.
\newblock Short interactive oracle proofs with constant query complexity, via
  composition and sumcheck.
\newblock {\em {ECCC}}, 23:46, 2016.

\bibitem[BEKP15]{BonnetE0P15}
Edouard Bonnet, Bruno Escoffier, Eun~Jung Kim, and Vangelis~Th. Paschos.
\newblock On subexponential and {FPT}-time inapproximability.
\newblock {\em Algorithmica}, 71(3):541--565, 2015.

\bibitem[BGKL03]{BGKL03}
L{\'{a}}szl{\'{o}} Babai, Anna G{\'{a}}l, Peter~G. Kimmel, and Satyanarayana~V.
  Lokam.
\newblock Communication complexity of simultaneous messages.
\newblock {\em {SIAM} J. Comput.}, 33(1):137--166, 2003.

\bibitem[BJKS04]{BJKS04}
Ziv Bar{-}Yossef, T.~S. Jayram, Ravi Kumar, and D.~Sivakumar.
\newblock An information statistics approach to data stream and communication
  complexity.
\newblock {\em J. Comput. Syst. Sci.}, 68(4):702--732, 2004.

\bibitem[BKK{\etalchar{+}}16]{BKKMS16}
Eli Ben{-}Sasson, Yohay Kaplan, Swastik Kopparty, Or~Meir, and Henning
  Stichtenoth.
\newblock Constant rate {PCP}s for circuit-sat with sublinear query complexity.
\newblock {\em J. {ACM}}, 63(4):32:1--32:57, 2016.

\bibitem[CCK{\etalchar{+}}17]{ChalermsookCKLMNT17}
Parinya Chalermsook, Marek Cygan, Guy Kortsarz, Bundit Laekhanukit, Pasin
  Manurangsi, Danupon Nanongkai, and Luca Trevisan.
\newblock From {G}ap-{ETH} to {FPT}-inapproximability: {C}lique, {D}ominating
  {S}et, and more.
\newblock In {\em FOCS}, pages 743--754, 2017.

\bibitem[CFK{\etalchar{+}}15]{CyganFKLMPPS15}
Marek Cygan, Fedor~V. Fomin, Lukasz Kowalik, Daniel Lokshtanov, D{\'{a}}niel
  Marx, Marcin Pilipczuk, Michal Pilipczuk, and Saket Saurabh.
\newblock {\em Parameterized Algorithms}.
\newblock Springer, 2015.

\bibitem[CFL83]{CFL83}
Ashok~K. Chandra, Merrick~L. Furst, and Richard~J. Lipton.
\newblock Multi-party protocols.
\newblock In {\em STOC}, pages 94--99, 1983.

\bibitem[CGG06]{ChenGG06}
Yijia Chen, Martin Grohe, and Magdalena Gr{\"{u}}ber.
\newblock On parameterized approximability.
\newblock In {\em IWPEC}, pages 109--120, 2006.

\bibitem[CH10]{CaiH10}
Liming Cai and Xiuzhen Huang.
\newblock Fixed-parameter approximation: Conceptual framework and
  approximability results.
\newblock {\em Algorithmica}, 57(2):398--412, 2010.

\bibitem[CHK13]{ChitnisHK13}
Rajesh~Hemant Chitnis, MohammadTaghi Hajiaghayi, and Guy Kortsarz.
\newblock Fixed-parameter and approximation algorithms: {A} new look.
\newblock In {\em IPEC}, pages 110--122, 2013.

\bibitem[CHKX06]{ChenHKX06}
Jianer Chen, Xiuzhen Huang, Iyad~A. Kanj, and Ge~Xia.
\newblock Strong computational lower bounds via parameterized complexity.
\newblock {\em J. Comput. Syst. Sci.}, 72(8):1346--1367, 2006.

\bibitem[Chv79]{Chvatal79}
Vasek Chv{\'{a}}tal.
\newblock A greedy heuristic for the set-covering problem.
\newblock {\em Math. Oper. Res.}, 4(3):233--235, 1979.

\bibitem[CL16]{ChenL16}
Yijia Chen and Bingkai Lin.
\newblock The constant inapproximability of the parameterized dominating set
  problem.
\newblock In {\em FOCS}, pages 505--514, 2016.

\bibitem[CMY08]{CMY08}
Graham Cormode, S.~Muthukrishnan, and Ke~Yi.
\newblock Algorithms for distributed functional monitoring.
\newblock In {\em SODA}, pages 1076--1085, 2008.

\bibitem[CNW16]{CharikarNW16}
Moses Charikar, Yonatan Naamad, and Anthony Wirth.
\newblock On approximating target set selection.
\newblock In {\em APPROX}, pages 4:1--4:16, 2016.

\bibitem[CRR14]{CRR14}
Arkadev Chattopadhyay, Jaikumar Radhakrishnan, and Atri Rudra.
\newblock Topology matters in communication.
\newblock In {\em FOCS}, pages 631--640, 2014.

\bibitem[DF95a]{DowneyF95}
Rodney~G. Downey and Michael~R. Fellows.
\newblock Fixed-parameter tractability and completeness {II:} on completeness
  for {W[1]}.
\newblock {\em Theor. Comput. Sci.}, 141(1{\&}2):109--131, 1995.

\bibitem[DF95b]{DowneyF1995}
Rodney~G. Downey and Michael~R. Fellows.
\newblock {\em Parameterized Computational Feasibility}, pages 219--244.
\newblock Birkh{\"a}user Boston, Boston, MA, 1995.

\bibitem[DF13]{DowneyF13}
Rodney~G. Downey and Michael~R. Fellows.
\newblock {\em Fundamentals of Parameterized Complexity}.
\newblock Texts in Computer Science. Springer, 2013.

\bibitem[DFM06]{DowneyFM06}
Rodney~G. Downey, Michael~R. Fellows, and Catherine McCartin.
\newblock Parameterized approximation problems.
\newblock In {\em IWPEC}, pages 121--129, 2006.

\bibitem[DFMR08]{DowneyFMR08}
Rodney~G. Downey, Michael~R. Fellows, Catherine McCartin, and Frances~A.
  Rosamond.
\newblock Parameterized approximation of dominating set problems.
\newblock {\em Inf. Process. Lett.}, 109(1):68--70, 2008.

\bibitem[Din07]{Dinur07}
Irit Dinur.
\newblock The {PCP} theorem by gap amplification.
\newblock {\em J. {ACM}}, 54(3):12, 2007.

\bibitem[Din16]{D16}
Irit Dinur.
\newblock Mildly exponential reduction from gap {3SAT} to polynomial-gap
  label-cover.
\newblock {\em ECCC}, 23:128, 2016.

\bibitem[DS14]{DinurS14}
Irit Dinur and David Steurer.
\newblock Analytical approach to parallel repetition.
\newblock In {\em STOC}, pages 624--633, 2014.

\bibitem[EG04]{EisenbrandG04}
Friedrich Eisenbrand and Fabrizio Grandoni.
\newblock On the complexity of fixed parameter clique and dominating set.
\newblock {\em Theor. Comput. Sci.}, 326(1-3):57--67, 2004.

\bibitem[Fei98]{Feige98}
Uriel Feige.
\newblock A threshold of ln \emph{n} for approximating set cover.
\newblock {\em J. {ACM}}, 45(4):634--652, 1998.

\bibitem[FOZ16]{FOZ16}
Orr Fischer, Rotem Oshman, and Uri Zwick.
\newblock Public vs. private randomness in simultaneous multi-party
  communication complexity.
\newblock In {\em SIROCCO}, pages 60--74, 2016.

\bibitem[GO95]{GajentaanO95}
Anka Gajentaan and Mark~H. Overmars.
\newblock On a class of o(n2) problems in computational geometry.
\newblock {\em Comput. Geom.}, 5:165--185, 1995.

\bibitem[Gol08]{G09}
Oded Goldreich.
\newblock {\em Computational Complexity: A Conceptual Perspective}.
\newblock Cambridge University Press, New York, NY, USA, 1 edition, 2008.

\bibitem[GR18a]{GR18a}
Oded Goldreich and Guy~N. Rothblum.
\newblock Counting {\textdollar}t{\textdollar}-cliques: Worst-case to
  average-case reductions and direct interactive proof systems.
\newblock {\em Electronic Colloquium on Computational Complexity {(ECCC)}},
  25:46, 2018.

\bibitem[GR18b]{GR18}
Oded Goldreich and Guy~N. Rothblum.
\newblock Simple doubly-efficient interactive proof systems for
  locally-characterizable sets.
\newblock In {\em 9th Innovations in Theoretical Computer Science Conference,
  {ITCS} 2018, January 11-14, 2018, Cambridge, MA, {USA}}, pages 18:1--18:19,
  2018.

\bibitem[GS96]{GS96}
Arnaldo Garcia and Henning Stichtenoth.
\newblock On the asymptotic behaviour of some towers of function fields over
  finite fields.
\newblock {\em Journal of Number Theory}, 61(2):248 -- 273, 1996.

\bibitem[HKK13]{HajiaghayiKK13}
Mohammad~Taghi Hajiaghayi, Rohit Khandekar, and Guy Kortsarz.
\newblock The foundations of fixed parameter inapproximability.
\newblock {\em CoRR}, abs/1310.2711, 2013.

\bibitem[IP01]{IP01}
Russell Impagliazzo and Ramamohan Paturi.
\newblock On the complexity of k-{SAT}.
\newblock {\em J. Comput. Syst. Sci.}, 62(2):367--375, 2001.

\bibitem[IPZ01]{IPZ01}
Russell Impagliazzo, Ramamohan Paturi, and Francis Zane.
\newblock Which problems have strongly exponential complexity?
\newblock {\em J. Comput. Syst. Sci.}, 63(4):512--530, 2001.

\bibitem[Joh74]{Johnson74a}
David~S. Johnson.
\newblock Approximation algorithms for combinatorial problems.
\newblock {\em J. Comput. Syst. Sci.}, 9(3):256--278, 1974.

\bibitem[Jus72]{Justesen72}
J{\o}rn Justesen.
\newblock Class of constructive asymptotically good algebraic codes.
\newblock {\em {IEEE} Trans. Information Theory}, 18(5):652--656, 1972.

\bibitem[Kar72]{Karp72}
Richard~M. Karp.
\newblock Reducibility among combinatorial problems.
\newblock In {\em Proceedings of a symposium on the Complexity of Computer
  Computations}, pages 85--103, 1972.

\bibitem[KN97]{KushilevitzNisan96-CCBook}
Eyal Kushilevitz and Noam Nisan.
\newblock {\em Communication Complexity}.
\newblock Cambridge University Press, New York, NY, USA, 1997.

\bibitem[KS92]{KS92}
Bala Kalyanasundaram and Georg Schnitger.
\newblock The probabilistic communication complexity of set intersection.
\newblock {\em {SIAM} J. Discrete Math.}, 5(4):545--557, 1992.

\bibitem[Lin15]{Lin15}
Bingkai Lin.
\newblock The parameterized complexity of \emph{k}-biclique.
\newblock In {\em SODA}, pages 605--615, 2015.

\bibitem[Lov75]{Lovasz75}
L.~Lov{\'{a}}sz.
\newblock On the ratio of optimal integral and fractional covers.
\newblock {\em Discrete Mathematics}, 13(4):383--390, 1975.

\bibitem[LS09]{LS09}
Troy Lee and Adi Shraibman.
\newblock Lower bounds in communication complexity.
\newblock {\em Foundations and Trends in Theoretical Computer Science},
  3(4):263--398, 2009.

\bibitem[LV11]{LV11}
Guanfeng Liang and Nitin~H. Vaidya.
\newblock Multiparty equality function computation in networks with
  point-to-point links.
\newblock In {\em SIROCCO}, pages 258--269, 2011.

\bibitem[LY94]{LundY94}
Carsten Lund and Mihalis Yannakakis.
\newblock On the hardness of approximating minimization problems.
\newblock {\em J. {ACM}}, 41(5):960--981, 1994.

\bibitem[Mei13]{M13}
Or~Meir.
\newblock {IP} = {PSPACE} using error-correcting codes.
\newblock {\em {SIAM} J. Comput.}, 42(1):380--403, 2013.

\bibitem[Mos15]{Moshkovitz15}
Dana Moshkovitz.
\newblock The projection games conjecture and the {NP}-hardness of ln
  n-approximating set-cover.
\newblock {\em Theory of Computing}, 11:221--235, 2015.

\bibitem[MR16]{MR16}
Pasin Manurangsi and Prasad Raghavendra.
\newblock A birthday repetition theorem and complexity of approximating dense
  {CSP}s.
\newblock {\em CoRR}, abs/1607.02986, 2016.

\bibitem[Nis94]{Nisan94}
Noam Nisan.
\newblock The communication complexity of threshold gates.
\newblock In {\em Proceedings of “Combinatorics, Paul Erdos is Eighty}, pages
  301--315, 1994.

\bibitem[Pat10]{P10}
Mihai Patrascu.
\newblock Towards polynomial lower bounds for dynamic problems.
\newblock In {\em STOC}, pages 603--610, 2010.

\bibitem[PVZ12]{PVZ12}
Jeff~M. Phillips, Elad Verbin, and Qin Zhang.
\newblock Lower bounds for number-in-hand multiparty communication complexity,
  made easy.
\newblock In {\em SODA}, pages 486--501, 2012.

\bibitem[PW10]{PatrascuW10}
Mihai Patrascu and Ryan Williams.
\newblock On the possibility of faster {SAT} algorithms.
\newblock In {\em SODA}, pages 1065--1075, 2010.

\bibitem[PY91]{PY91}
Christos~H. Papadimitriou and Mihalis Yannakakis.
\newblock Optimization, approximation, and complexity classes.
\newblock {\em J. Comput. Syst. Sci.}, 43(3):425--440, 1991.

\bibitem[Raz92]{R92}
Alexander~A. Razborov.
\newblock On the distributional complexity of disjointness.
\newblock {\em Theor. Comput. Sci.}, 106(2):385--390, 1992.

\bibitem[Raz98]{Raz98}
Ran Raz.
\newblock A parallel repetition theorem.
\newblock {\em {SIAM} J. Comput.}, 27(3):763--803, 1998.

\bibitem[RS97]{RazS97}
Ran Raz and Shmuel Safra.
\newblock A sub-constant error-probability low-degree test, and a sub-constant
  error-probability {PCP} characterization of {NP}.
\newblock In {\em STOC}, pages 475--484, 1997.

\bibitem[Rub17]{Rub17com}
Aviad Rubinstein.
\newblock Personal communication, 2017.

\bibitem[Rub18]{R17}
Aviad Rubinstein.
\newblock Hardness of approximate nearest neighbor search.
\newblock In {\em STOC}, 2018.
\newblock To appear.

\bibitem[SAK{\etalchar{+}}01]{SAKSD01}
Kenneth~W. Shum, Ilia Aleshnikov, P.~Vijay Kumar, Henning Stichtenoth, and
  Vinay Deolalikar.
\newblock A low-complexity algorithm for the construction of
  algebraic-geometric codes better than the {Gilbert-Varshamov} bound.
\newblock {\em {IEEE} Trans. Information Theory}, 47(6):2225--2241, 2001.

\bibitem[Sla96]{Slavik96}
Petr Slav{\'{\i}}k.
\newblock A tight analysis of the greedy algorithm for set cover.
\newblock In {\em STOC}, pages 435--441, 1996.

\bibitem[Sri95]{Srinivasan95}
Aravind Srinivasan.
\newblock Improved approximations of packing and covering problems.
\newblock In {\em STOC}, pages 268--276, 1995.

\bibitem[SS96]{SipserS96}
Michael Sipser and Daniel~A. Spielman.
\newblock Expander codes.
\newblock {\em {IEEE} Trans. Information Theory}, 42(6):1710--1722, 1996.

\bibitem[Tov84]{Tovey84}
Craig~A. Tovey.
\newblock A simplified {NP}-complete satisfiability problem.
\newblock {\em Discrete Applied Mathematics}, 8(1):85--89, 1984.

\bibitem[Vio15]{Viola15}
Emanuele Viola.
\newblock The communication complexity of addition.
\newblock {\em Combinatorica}, 35(6):703--747, 2015.

\bibitem[Wil05]{W05}
Ryan Williams.
\newblock A new algorithm for optimal 2-constraint satisfaction and its
  implications.
\newblock {\em Theor. Comput. Sci.}, 348(2-3):357--365, 2005.

\bibitem[WW15]{WW15}
Omri Weinstein and David~P. Woodruff.
\newblock The simultaneous communication of disjointness with applications to
  data streams.
\newblock In {\em ICALP}, pages 1082--1093, 2015.

\bibitem[Yao79]{Yao79}
Andrew~Chi{-}Chih Yao.
\newblock Some complexity questions related to distributive computing
  (preliminary report).
\newblock In {\em STOC}, pages 209--213, 1979.

\end{thebibliography}
\appendix

\section{A Reduction From \maxcover to \domset}
\label{app:maxcov-domset}

In this section, we provide the reduction from \maxcover to \domset (i.e. Theorem~\ref{lem:red-maxcov-domset}). The reduction and its proof presented here are quoted almost directly from Chalermsook et al.~\cite{ChalermsookCKLMNT17} except with appropriate notational adjustments. We include it here only for the sake of self-containedness.

To ease the presentation, we will present the reduction in terms of a different variant of the label cover problem called \minlabel. The input to \minlabel is the same as that of \maxcover, i.e., it is a label cover instance $\Gamma = (U = \bigcup_{i\in[q]} U_i, W = \bigcup_{j\in[k]} W_j; E)$ as defined in Subsection~\ref{sec:prelim:label-cover}. 

However, the solution and the objective of the problem will be different from \maxcover. Specifically, a solution of \minlabel is called a {\em multi-labeling}, which is simply a subset of vertices $\tS \subseteq W$.
We say that a multi-labeling $\tS$ {\em covers} a left super-node $U_{i}$ if there exists a vertex $u_{i} \in U_{i}$ that has a neighbor in $\tS \cap W_j$ for every $j \in [k]$. The goal in \minlabel is to find a minimum-size multi-labeling that covers all the left super-nodes $U_{i}$'s. Again, we overload the notation and use $\minlabel$ to also denote the optimum, i.e., 
\begin{align*}
\minlabel(\Gamma) = \min_{\substack{\text{multi-labeling } \tS \\ \text{ that covers every $U_i$}}} |\tS|.
\end{align*}

%Since \maxcover and \minlabel are problems on the same input instance, we have the following connection between \maxcover and \minlabel, which allows us to deduce inapproximability results of \minlabel from that of \maxcover.

The following lemma demonstrates relations between \maxcover and \minlabel on the same input label cover instance. These relations in turns allow us to deduce inapproximability results of \minlabel from that of \maxcover.

\begin{proposition}
\label{Proposition:maxcover-vs-minlabel}
Let $\Gamma$ be any label cover instance with $k$ right super-nodes. Then, we have
\begin{itemize}
\item If $\maxcover(\Gamma) = 1$, then $\minlabel(\Gamma)=k$.
\item If $\maxcover(\Gamma) \leq \epsilon$ for some $0< \epsilon$, then $\minlabel(\Gamma) \geq (1/\epsilon)^{1/k}\cdot k$.
\end{itemize}
\end{proposition}

%The proof of Proposition~\ref{Proposition:maxcover-vs-minlabel} follows from a standard technique (see, e.g., \cite{ChalermsookCKLMNT17}).
%We provide the proof in Appendix~\ref{app:maxcover-vs-minlabel}.

Before proceeding to prove Proposition~\ref{Proposition:maxcover-vs-minlabel}, let us state the desired reduction from label cover to \domset (Theorem~\ref{lem:red-maxcov-domset}) in terms of \minlabel instead of \maxcover:

\begin{lem}[Reduction from \minlabel to \domset~{\cite[Theorem 5.4]{ChalermsookCKLMNT17}})] \label{lem:red-minlab-domset}
There is an algorithm that, given a label cover instance 
$\Gamma = (U = \bigcup_{j=1}^{q} U_j, W = \bigcup_{i=1}^{k} W_i, E)$,
outputs a $k$-\domset instance 
$G$ such that
\begin{itemize}
\item $\minlabel(\Gamma) = \domset(G)$.
\item $|V(G)| = |W| + \sum_{i \in [q]} k^{|U_q|}$.
\item The reduction runs in time $O\left(|W| \left(\sum_{i \in [q]} k^{|U_q|}\right)\right)$.
\end{itemize}
\end{lem}

It is obvious that Proposition~\ref{Proposition:maxcover-vs-minlabel} and Lemma~\ref{lem:red-minlab-domset} together imply Theorem~\ref{lem:red-maxcov-domset}. The rest of this section is devoted to proving Proposition~\ref{Proposition:maxcover-vs-minlabel} and Lemma~\ref{lem:red-minlab-domset}. Their proofs can be found in Subsections~\ref{app:maxcover-vs-minlabel} and~\ref{apx:minlabel-to-domset}, respectively.

\subsection{From \maxcover to \minlabel: Proof of Proposition~\ref{Proposition:maxcover-vs-minlabel}} \label{app:maxcover-vs-minlabel}

In this section, we prove Proposition~\ref{Proposition:maxcover-vs-minlabel}. The proof presented here is due to~\cite{ChalermsookCKLMNT17} with slight changes to match our notations.

\begin{proof}[Proof of Proposition~\ref{Proposition:maxcover-vs-minlabel}]
\begin{enumerate}[(1)]
\item Suppose $\maxcover(\Gamma) = 1$. Then there exists a labeling $S \subseteq W$ that covers every left super-node $U_i$. Hence, $S$ is also a multi-labeling that covers every left super-node, which implies that $\minlabel(\Gamma) = k$.
\item We prove by contrapositive. Assume that $\minlabel(\Gamma) < (1/\varepsilon)^{1/k} k$. Then there exists a multi-labeling $\tS \subseteq W$ of size less than $(1/\varepsilon)^{1/k} k$ that covers every left super-node. Observe that we can construct a feasible $S'$ solution to $\maxcover$ by uniformly and independently choosing one node from $\tS\cap W_j$ at random, for each of the right super-node $W_j$.
  
  Thus, the expected number of left super-nodes covered by $S'$ is \allowdisplaybreaks

  \begin{align*}
   \mathbb{E}_{S'}\left[|\{i\in[q]: S' \mbox{ covers } U_i\}|\right]
      &\ge \sum_{i\in[q]} \prod_{j\in[k]} |\tS \cap W_j|^{-1}\\
      \mbox{(By AM-GM inequality)} &\geq \sum_{i\in[q]} \left(\frac{1}{k}\sum_{j\in[k]}|\tS\cap W_j|\right)^{-k}\\
      &> q \cdot \left(\frac{1}{k}\cdot\left(\left(\frac{1}{\varepsilon}\right)^{1/k}k\right)\right)^{-k}\\
      & = q \cdot \varepsilon
  \end{align*}
  Hence, there is a labeling that covers $> \varepsilon q$ left super-nodes, i.e., $\maxcover(\Gamma) > \varepsilon$.
  \qedhere
\end{enumerate}
\end{proof}

\subsection{From \minlabel to \domset: Proof of Lemma~\ref{lem:red-minlab-domset}}
\label{apx:minlabel-to-domset}

We devote this subsection to show a reduction from \minlabel to \domset
as stated in Lemma~\ref{lem:red-minlab-domset}.

\begin{proof}[Proof of Lemma~\ref{lem:red-minlab-domset}]
We show the reduction from \minlabel to the bipartite version of \domset,
where the input graph $G=(A, B; E)$ is bipartite,
and we are asked to find a subset of (left) vertices of $A$ that dominates
all the (right) vertices in $B$. 
This variant of \domset is known as {\sf Red-Blue} $\domset$, 
which is equivalent to the {\em Set Cover} problem.
%,where we wish to pick a collection of subsets so that the union contains all the elements.
%
One can transform the bipartite version of \domset into the general case just 
by forming a clique on $A$. 
It is not hard to see that %any optimal solution to \domset contains only vertices from $A$,
%and thus, the two instances share the same optimal solution.
the optimum remains unchanged.

We apply the reduction from \cite{ChalermsookCKLMNT17},
which is in turn taken from \cite{Feige98} (and \cite{LundY94})
with small adjustments.
%% a slight modification of that by Feige~\cite{Feige98}
%% and by Lund-Yannakakis~\cite{LundY94}.
%
Our construction requires the {\em hypercube partition system}.

\paragraph{Hypercube Partition System.}
The $(\kappa,\rho)$-hypercube partition system consists of the universe $\mathcal{M}$
and a collection of subset $\{P_{x,y}\}_{x\in [\rho],y\in[\kappa]}$ where
\begin{align*}
\mathcal{M} = [\kappa]^{\rho}
\quad\mbox{and}\quad
P_{x,y} = \{z\in \mathcal{M}: z_{x} = y\}
\end{align*}

The following are properties that can be observed from the hypercube partition system:
\begin{itemize}
\item {\bf Partition Cover:} For each row $x\in[\rho]$, 
      $\mathcal{P}_x=(P_{x,1},\ldots,P_{x,\kappa})$ is a partition of $\mathcal{M}$.
      Thus, $\bigcup_{y\in[\kappa]}P_{x,y}=\mathcal{M}$ and 
      $P_{x,y} \cap P_{x,y'} = \emptyset$ for $y\neq y'$.

      \medskip
      That is, ones can cover the universe $\mathcal{M}$ by picking 
      all $\kappa$ subsets from some row $x\in[\kappa]$.

\item {\bf Robust Property:} 
      For any $y^*_{1},\ldots,y^*_{\rho}\in[\kappa]$, 
      $\bigcup_{x\in[\alpha], y\in[\kappa]\setminus \{y^*_x\}} P_{x, y} \neq \mathcal{M}$.

      \medskip

      That is, the only way to cover the universe $\mathcal{M}$ is to 
      pick all the $\kappa$ subsets from the same row.
      Otherwise, even if we include $\kappa-1$ subsets from 
      every row $x\in[\rho]$, 
      it is not possible to cover the universe $\mathcal{M}$.

\item {\bf Size:} The number of elements in the partition is $\kappa^\rho$.
\end{itemize}

\paragraph*{The Reduction.}
Now we present our reduction.
Let $\Gamma=(U=\bigcup_{j=1}^{q}U_j, W=\bigcup_{i=1}^k W_i; E')$
be an instance of \minlabel.
We will construct a bipartite graph $G=(A,B; E)$ of {\sf Red-Blue} \domset from $\Gamma$.
The vertices of $A$ are taken from the right vertices of $\Gamma$, i.e., $A=W$.

For each super-node $U_{i}$, we create the $(k,|U_{i}|)$-hypercube partition system
$(\mathcal{M}^{i},\{P^{i}_{x,y}\}_{x\in[|U_{q}|],y\in[k]})$
(i.e., $\rho=|U_{q}|$ and $\kappa=k$);
then we take the elements of $\mathcal{M}^j$ as left vertices of $G$.

We name vertices in $U_i$ by $u^i_1,u^i_2,\ldots,u^i_{|U_i|}$ for $i\in[q]$.
% and name vertices in $W_j$ by $w^j_1,w^j_2,\ldots,w^j_{|W_j|}$ for $j\in[k]$.
%
For each $j\in k$ and each $w^j\in W_j$,
we join $w_j$ to every vertex $z\in P^i_{x,j}$ 
if $\{u^i_x, w^j\}$ is an edge in $\Gamma$.
More precisely, the bipartite graph $G=(A,B; E)$ is defined such that 
\begin{align*}
A := W = \bigcup_{j}^{k}W_j,
B := \bigcup_{i=1}^{q} \mathcal{M}^i,
E := \{\{w^j, z\}: i \in [q], j\in[k], x\in[|U_i|], w^j\in W_j,
              z \in P^i_{x,j}, \{u^i_x, w^j\} \in E'\}.
\end{align*}

\paragraph*{Size of The Construction.}
It is clear from that the numbers of left and right vertices in our construction is as follows.
%Let $\sigma_{\text{left}}=\max_{i\in[q]} |U_i|$ and 
%$\sigma_{\text{right}}=\max_{j\in[k]} |W_j|$
%be the left-alphabet size and the right-alphabet size of $\Gamma$, respectively.

\begin{align*}
\mbox{The Number of Left veritces of $G$} 
  &= |W| \\
  %\leq k\cdot\sigma_{\text{right}}\\
%
\mbox{The Number of Right vertices of $G$}
  &= \sum_{i\in[q]}k^{|U_i|}
  %\leq q\cdot k^{\sigma_{\text{left}}}\\
%
%\mbox{Instance Size} 
%  &\leq k^{\Sigma_{\text{left}}} + k\cdot 2^{\sigma_{\text{right}}}
\end{align*}

\paragraph*{Analysis.}
Clearly, the size translation satisfies the lemma statement. 
For the claim $\minlabel(\Gamma)=\domset(G)$,
observe that any solution $S$ to the \minlabel instance corresponds to 
a solution to the {\sf Red-Blue} \domset instance.
We claim that $S$ is feasible to the \minlabel instance if and only if 
it is feasible to the \domset instance.
This will prove the claim and thus imply the lemma.

First, suppose $S$ is feasible to \minlabel.
Then, for every $U_i$, there is a vertex $u^i_x\in U_i$ that has an edge
to some vertex $w^j_y\in W_j\cap S$, for all $j\in[k]$.
This means that we have an edge from $w^j_y$ to every vertex $z\in P^i_{x,y}$.
It then follows by the Partition Cover property of
the hypercube partition system that $S$ dominates 
every vertex in the universe $\mathcal{M}^i$.
Consequently, we conclude that $S$ dominates $B=\bigcup_{i\in[q]}\mathcal{M}^i$.

Conversely, suppose $S$ is feasible to \domset.
Then $S$ has an edge to every vertex in the universe $\mathcal{M}^i\subseteq B$.
We know by construction and 
the Robustness Property of the hypercube partition system
that, for some $x\in [|U_i|]$,
$S$ contains vertices 
$w^{1} \in S\cap W_1, w^{2} \in S\cap W_2, \ldots, w^{k} \in S\cap W_{k}$,
which correspond to the subsets $P^i_{x,y}$'s in the partition system,
and that $w^{1},w^{2},\ldots,w^{k}$ are incident to the vertex $u^i_x \in U_i$.
This means that $S$ covers the left super-node $U_i$.
Consequently, we conclude that $S$ is a feasible multi-labeling,
thus completing the proof.
\end{proof}
\section{Connection to Fine-Grained Complexity}\label{sec:fine}

In this section, we will demonstrate the conditional hardness of problems in \P\ by basing them on the conditional hardness of \PSP. First, we define the problem in \P\ of interest to this section.

\begin{definition}[$k$-linear form inner product]
Let $k,m\in\mathbb N$. Given $x_1,\ldots ,x_k\in\mathbb R^m$, we define the inner product of these $k$ vectors as follows:
$$\langle x_1,\ldots ,x_k\rangle=\sum_{i\in[m]}\prod_{j\in[k]}x_i(j),$$
where $x_i(j)$ denotes the $j^{\text{th}}$ coordinate of the vector $x_i$.
\end{definition}

\begin{definition}[$k$-chromatic Maximum Inner Product]
Let $k\in \N$.  Given $k$ collections $A_1,A_2,\ldots ,A_k$ each of $N$ vectors in $\{0,1\}^D$, where $D=N^{o(1)}$, and an integer $s$, the $k$-chromatic Maximum Inner Product (\mip) problem is to determine if there exists $a_i\in A_i$  for all $i\in[k]$ such that $\langle a_1,\ldots ,a_k\rangle\ge s$.
\end{definition}

We continue to use the shorthand $\Gamma(N,k,r,\ell)$ introduced in Section~\ref{sec:gapmax}. Moreover, while using the shorthand $\Gamma(N,k,r,\ell)$, we will assume that for all $i\in [h]$, $|W_i|\le N$. 
We define {\sf Unique} \maxcover to be the \maxcover problem with the following additional structure: for every labeling $S\subseteq W$ (see Section~\ref{sec:prelim:label-cover} to recall the definition of a labeling) and any left super-node $U_i$, there is at most one node in $U_i$ which is a neighbor to all the nodes in $S$.
We remark that the reduction in Theorem~{\ref{thm:meta-maxcov}} already produces instances of {\sf Unique} \maxcover. This follows from the proof of Theorem~\ref{thm:meta-maxcov} by noting that on every random string, each player sends a message to the referee in a deterministic way. 
Finally, we have the following connection between unique \maxcover and $\mip$. 

\begin{theorem}\label{thm:metafine}
Let $N,k,r,\ell\in\mathbb N$.
There is a reduction from any {\sf Unique} \maxcover instance $\Gamma(N,k,r,\ell)$ to $k$-chromatic \mip instance $(A_1,\ldots ,A_k,s)$ such that
\begin{itemize}
\item For all $i\in [k]$, $|A_i|\le N$, $s=2^r$, and $D\le 2^{r+\ell k}$.
\item The running time of the reduction is $O(Nk\cdot 2^{r+\ell k}) $.
\item For any integer $s^*$, there exists $a_i\in A_i$  for all $i\in[k]$ such that $\langle a_1,\ldots ,a_k\rangle\ge s^*$ if and only if $\maxcover(\Gamma)\ge \nicefrac{s^*}{2^r}$.   
\end{itemize}
\end{theorem}
\begin{proof}
For every $i\in[k]$, we associate each $A_i$ with the right super-node $W_i$. Each node in $W_i$ corresponds to a vector in $A_i$. Given a node $w$ in $W_i$, the corresponding vector $a\in A_i$ is constructed as follows. Each coordinate of $a$ corresponds to a left node. Let  the $j^{\text{th}}$ coordinate of $a$ correspond to a node $u\in U$. The $j^{\text{th}}$ coordinate of $a$ is 1 if there is an edge between $u$ and $w$ in the graph $G$ of the instance $\Gamma$; otherwise, the $j^{\text{th}}$ coordinate of $a$ is 0. It is easy to see that $|A_i|\le N$ and $D\le 2^{r+\ell k}$, and the running time of the reduction is $O(Nk\cdot 2^{r+\ell k}) $. It is also easy to see that if $\maxcover(\Gamma)\ge \nicefrac{s^*}{2^r}$ then the vectors $a_i$ corresponding to the nodes in the labeling resulting in the maximum cover, have inner product at least $s^*$.   

We will now show that  for any integer $s^*$, if there exists $a_i\in A_i$  for all $i\in[k]$ such that $\langle a_1,\ldots ,a_k\rangle\ge s^*$ then $\maxcover(\Gamma)\ge \nicefrac{s^*}{2^r}$. Fix an integer $s^*$. Suppose that there exists $a_i\in A_i$  for all $i\in[k]$ such that $\langle a_1,\ldots ,a_k\rangle\ge s^*$. The $k$-tuple $(a_1,\ldots ,a_k)$ corresponds to a labeling $S$ of unique \maxcover. For each coordinate such that all the $a_i=1$ on that coordinate, we note that there is a left node (corresponding to the coordinate) which is a common neighbor of all the nodes in the labeling. Moreover, since in the unique \maxcover we can have at most one node in each left super-node which is a common neighbor of the nodes in the labeling, we have that:
 $$\maxcover(\Gamma)\ge \frac{1}{2^r}\cdot |\{i \in [2^r] \mid U_i \text{ is covered by } S\}|=\langle a_1,\ldots ,a_k\rangle\ge s^*.\qedhere$$
\end{proof}

The results for hardness of approximation for problems in \P\ is obtained by simply fixing the value of $k$ in the above theorem to some universal constant (such as $k=2$).
For example, by fixing $k=2$ and applying the above reduction to Corollary~\ref{cor:labSETH}, we recover the main result of \cite{ARW17} on bichromatic \mip. This is not surprising as under \SETH, our framework is just a generalization of the distributed PCP framework of \cite{ARW17}.

Furthermore, we demonstrate the flexibility of our framework by fixing $k=3$ and applying the above reduction to Corollary~\ref{cor:labksum}, to establish a hardness of approximation result of trichromatic \mip\ under the \textsf{3-SUM} Hypothesis as stated below. We note here that the running time lower bound is only $N^{2 - o(1)}$, which is  likely not tight since the lower bound from \SETH is $N^{3 - o(1)}$.

\begin{theorem}\label{thm:3sumMIP}
Assuming the \textsf{3-SUM} Hypothesis, for every $\varepsilon>0$,
no $O(N^{2-\varepsilon})$ time algorithm can,
given three collections $A,B,$ and $C$ each of $N$ vectors in $\{0,1\}^D$, where $D=N^{o(1)}$, and an integer $s$, distinguish between the following two cases:
\begin{description}
\item[Completeness.] There exists $a\in A,b\in B,c\in C$ such that $\langle a,b,c\rangle\ge s$.
\item[Soundness.] For every $a\in A,b\in B,c\in C$, $\langle a,b,c\rangle\le s/2^{(\log N)^{1-o(1)}}$.
\end{description}
\end{theorem}
\begin{proof}
We apply Proposition~\ref{prop:repeat} with $z=\nicefrac{m}{(\log_2 m)^2}$ and $k=3$ to Theorem~\ref{thm:RandKSum}, to obtain a  $\left(0,O(\nicefrac{m}{\log m}),O(\nicefrac{m}{(\log m)^2}),\left(\nicefrac{1}{2}\right)^{m^{1-o(1)}}\right)$-efficient protocol for $\SZ_{m, 3}$ in the $3$-player \SMP model. By plugging in the parameters of the above protocol to Corollary~\ref{cor:labksum}, we obtain that assuming the \textsf{3-SUM} hypothesis, for every $\varepsilon > 0$, no $O(N^{2 - \varepsilon})$-time algorithm can distinguish between $\maxcover(\Gamma) = 1$
and $\maxcover(\Gamma) \leq \left(\nicefrac{1}{2}\right)^{(\log N)^{1-o(1)}}$ for any label cover instance $\Gamma(N,3,r,\ell)$ for all $N \in \N$. The proof of the theorem concludes by applying Theorem~\ref{thm:metafine} to the above hardness of \maxcover (Note that Theorem~\ref{thm:meta-maxcov} provides a reduction from $\PSP(k,\cF,N)$ to {\sf Unique} \maxcover).
\end{proof}

\end{document}